\documentclass[11pt]{amsart}
\usepackage{amsbsy,amssymb,amscd,amsfonts,latexsym,amstext,delarray,
amsmath,graphicx,slashed} 
\usepackage[all,cmtip]{xy}
\usepackage{fullpage}
\usepackage[pdftex]{color}

\newtheorem{thm}{Theorem}[section]
\newtheorem{prop}[thm]{Proposition}
\newtheorem{cor}[thm]{Corollary}
\newtheorem{lem}[thm]{Lemma}

\newtheorem{rem}[thm]{Remark}

\numberwithin{equation}{section}

\def\A{{\mathbb A}}
\def\C{{\mathbb C}}

\def\N{{\mathbb N}}
\renewcommand{\P}{{\mathbb P}}
\def\Q{{\mathbb Q}}
\def\Z{{\mathbb Z}}
\def\R{{\mathbb R}}

\def\cA{{\mathcal A}}

\def\cD{{\mathcal D}}

\def\cF{{\mathcal F}}
\def\cG{{\mathcal G}}
\def\cH{{\mathcal H}}
\def\cI{{\mathcal I}}

\def\cK{{\mathcal K}}
\def\cL{{\mathcal L}}
\def\cM{{\mathcal M}}

\def\cR{{\mathcal R}}
\def\cS{{\mathcal S}}

\def\cV{{\mathcal V}}

\def\cX{{\mathcal X}}
\def\cY{{\mathcal Y}}

\def\m{{\mathfrak m}}
\def\fH{{\mathfrak H}}
\def\fg{{\mathfrak g}}

\title{Algebraic renormalization and Feynman integrals in configuration spaces}
\author{\"Ozg\"ur Ceyhan and Matilde Marcolli}
\address{Facult\'e des Sciences, de la Technologie et de la Communication, University of Luxembourg, 6 rue Richard Coudenhove-Kalergi, L-1359 Luxembourg}
\email{ozgur.ceyhan@uni.lu}
\address{Division of Physics, Mathematics and Astronomy,
Mail Code 253-37, Caltech, 1200 E.~California Blvd. Pasadena, CA 91125, USA}
\email{matilde@caltech.edu} 

\date{\today}
 
\begin{document}
\maketitle

\begin{abstract}
This paper continues our previous study of Feynman integrals in configuration
spaces and their algebro-geometric and motivic aspects. We consider here both 
massless and massive Feynman amplitudes, from the point of view of potential
theory. We consider a variant of the wonderful compactification of configuration
spaces that works simultaneously for all graphs with a given number of vertices
and that also accounts for the external structure of Feynman graph. As in our
previous work, we consider two version of the Feynman amplitude in configuration
space, which we refer to as the real and complex versions. In the real version,
we show that we can extend to the massive case a method of evaluating
Feynman integrals, based on expansion in Gegenbauer polynomials, that
we investigated previously in the massless case. In the complex setting,
we show that we can use algebro-geometric methods to renormalize the
Feynman amplitudes, so that the renormalized values of the Feynman integrals
are given by periods of a mixed Tate motive. The regularization and renormalization
procedure is based on pulling back the form to the wonderful compactification
and replace it with a cohomologous one with logarithmic poles. A complex
of forms with logarithmic poles, endowed with an operator of pole subtraction,
determine a Rota--Baxter algebra on the wonderful compactifications.
We can then apply the renormalization procedure via Birkhoff factorization,
after interpreting the regularization as an algebra homomorphism from the 
Connes--Kreimer Hopf algebra of Feynman graphs to the Rota--Baxter algebra.
We obtain in this setting a description of the renormalization group.
We also extend the period interpretation to the case of  Dirac fermions
and gauge bosons.
\end{abstract}

\tableofcontents

\section{Introduction}

In \cite{CeyMar1} and \cite{CeyMar2} we started a program aimed at studying the
relation between Feynman integrals in configuration space and periods and motives 
of algebraic varieties. We continue here this investigation by extending the previous
results in different directions. 

In \S \ref{ConfigSec} we introduce a variant of the configuration spaces
we worked with in the past. Instead of working here with a different
configuration space ${\rm Conf}_\Gamma(X)$ for each Feynman graph $\Gamma$,
we consider a single configuration space for all graphs with a fixed number of
vertices. The configuration spaces themselves, and their wonderful compactifications
$\cX_{\cH}[n]$,
are different from the setting we previously considered also because they
incorporate some assigned data of external structure of Feynman graphs,
given by imposing the condition that the configurations of points are away from
a prescribed locus $\cH$. This type of configuration spaces, with their compactifications
were previously considered by Kim and Sato, \cite{KS} and we simply we recall their 
main properties in our setting.

In \S \ref{FeySec} we describe the Feynman rules, the resulting Euclidean Feynman
amplitudes, and their unrenormalized Feynman integrals in configuration
spaces, distinguishing between the real and the complex case and between
the massless and the massive case.

In \S \ref{PotSec} we consider the massless and massive propagators in
configuration and momentum space and their description in terms of
potential theory, using Riesz and Bessel potentials and we justify in this
way more rigorously the definition of the Feynman amplitudes given in the
previous section. We also describe explicitly expansions of the massive
Feynman amplitudes based on expansions of modified Bessel functions
of the second type. 

In \S \ref{GegenSec} we focus on the massive real case. We extend to
the massive case 
the method we developed in \cite{CeyMar2} to compute massles
Euclidean Feynman integrals in the real case, using expansion of the
Feynman amplitude in Gegenbauer polynomials and a subdivision of
the chain of integration in loci constructed using acyclic orientations
of the Feynman graph. We use the expansion of the massive
Feynman amplitude obtained in the previous section using
expansions of the Bessel function, and we show that we can
further expand the terms obtained in this way in Gegenbauer
polynomials so that the same method developed in \cite{CeyMar2}
applies.

In \S \ref{RBSec} we consider instead the complex case. We show that
this case fits into the formalism of algebraic renormalization based on
Hopf algebras of Feynman graphs and Rota--Baxter algebras of weight
$-1$ as the target of the regularized Feynman amplitudes, or algebraic
Feynman rules. We consider a regularization procedure for the complex
Feynman amplitudes, similar to the case of \cite{CeyMar2}, obtained
by pulling back the smooth closed differential form given by the Feynman
amplitude, with singularities along the diagonals, to a form on the
wonderful compactification of the configuration space, with singularities
along the boundary divisor, and by further replacing it with a cohomologous
algebraic differential form with logarithmic poles along the same normal
crossings divisor. We then show that one can define an associative
algebra of even forms with logarithmic poles on which there is a 
linear operator of ``pole subtraction" that satisfies the Rota--Baxter relation
of weight $-1$. The regularized amplitude defines a homomorphism
of commutative algebras from the Connes--Kreimer Hopf algebra of
Feynman graphs to this Rota--Baxter algebra, hence we can renormalize
it using the Birkhoff factorization determined by the Hopf coproduct
and antipode, on one side, and the Rota--Baxter operator on the other.
We show that this renormalization method also provides a beta function
for the renormalization group.

In \S \ref{nonscalarSec} we show that the setting of \S \ref{FeySec}
and \S \ref{PotSec} can be generalized to include the cases of
Dirac fermions and gauge bosons.

\section{Configuration space of points}\label{ConfigSec}

Our  study of a algebro-geometric regularization procedures for  
position space Feynman integrals and their interpretation in terms of 
periods is based on wonderful compactifications  and
deformations of configuration space of points, where the
configuration spaces considered depend on the Feynman graphs
of the quantum field theory. 

The geometry of the wonderful compactifications of these graph configuration spaces   
was examined in detail in our previous work \cite{CeyMar1} and \cite{CeyMar2}, see
also \cite{Li1}, \cite{Li2} and \cite{Martin}.

In this paper, we change this point of view slightly and we consider a fixed configuration
space that corresponds to the case of the complete graph, which provides  an ``all in one"
tool to examine Feynman integrals associated to all possible graphs. This means that,
for any given Feynman amplitude, one will typically be performing more blowups than
what would be strictly necessary, but with the advantage of having a single space that accommodates
all the amplitudes at once. This configuration
space is closely related to the original Fulton--MacPherson case \cite{FM}, which also 
was based on the case of the complete graph, but the version we consider also
contains conditions aimed at encoding the data of external structure of Feynman graphs,
expressed in the form of certain ``separation conditions", as in Kim--Sato, \cite{KS}.
We recall here the main definitions and statements that we will need in the following.
The arguments are very similar to the case previously analyzed in  \cite{CeyMar2} and
\cite{KS}, so we do not reproduce them in full.

\medskip
\subsection{Configuration spaces with separation conditions}

We consider spaces describing configurations of $n$ distinct points on a
variety that are avoid a specified subvariety. 

\smallskip

Let $\cX$ a smooth projective variety and let $\cH = \bigcup_{c} \cH_{c}$ 
be a nonsingular closed proper subvariety of $\cX$ where  
$\cH_c$  are  irreducible  components of $\cH$ and 
$c \in \underline{k}^+ := \{1,\dots,k,\infty\} $.

The {\it configuration space} $\cF(\cX \setminus \cH,n)$ is the complement 
\begin{equation} \label{Zconfig}
(\cX \setminus \cH)^n 
\setminus 
\left ( 
\bigcup_{ \{i,j\} \subset \underline{n}} \Delta_{\{i,j\}}  
\right), 
\end{equation}
of the {\it  diagonals}
\begin{eqnarray}\label{Diag}
\Delta_{\{i,j\}} = \{( z_i \mid i \in \underline{n}) \in \cX^n \mid  z_i = z_j \}, 
\end{eqnarray}
associated to the pairs $\{i,j\}$ in $\underline{n} := \{1,\dots,n\}$ 
in the cartesian product $(\cX \setminus \cH)^n$.
 
\medskip 
\subsection{Wonderful compactification of configuration spaces}
We now describe the various steps involved in the construction of
wonderful compactifications for the configuration spaces $\cF(\cX \setminus \cH,n)$.

\medskip
\subsubsection{Diagonals and subvarieties}
We start by considering the following collection of subvarieties of $\cX^n$:

\begin{itemize}
\item Given a non-empty subset $S$ of $\underline{n}$ and 
a subset $\alpha$ of  $\underline{k}^+$,  
we have  a subvariety 
\begin{equation}\label{SubspInf}
\Delta_{\alpha,S} = \{( z_i \mid i \in \underline{n}) 
\mid \forall i \in S, \, \exists c\in \alpha \, \text{ with } \, z_i \in \cH_c \}
\end{equation} 
\item  For subsets $I \subset \underline{n}$ with
$|I|>1$, we have the corresponding diagonals 
\begin{equation}\label{DiagI}
\Delta_{I} = \{( z_i \mid i \in \underline{n}) \mid z_i = z_j \ \text{for} \  i,j \in I \}. 
\end{equation}
\end{itemize}

\medskip
\subsubsection{First step: arrangements and configurations separated from $\cH$}
The subvarieties   $\Delta_{\{c\},S}$ for $c \in \underline{k}^+$ are isomorphic to 
$$
 \Delta_{\{c\},S} \cong  \cH_c^S \times \cX^{\underline{n} \setminus S} .
$$
The collection $\cS_n = \{ \cap_{\alpha ,S} \Delta_{\alpha,S} \}$ of all possible intersections
of subvarieties $\Delta_{\alpha,S}$  forms a  simple arrangement in $\cX^n$. 

 \begin{lem}[Kim \& Sato \cite{KS}, \S 2.2]
The set $\cG_n = \{ \Delta_{\{c\},S} \mid c \in \underline{k}^+ \ \& \ S \subset \underline{n} \}$
is a building set for the arrangement $\cS_n$. 
\end{lem}

\medskip
\subsubsection{Blow-up construction of the first step in the compactification}

We introduce a partial ordering $\preceq$ on
$\cG_n$ according to the dimension of subspaces i.e., $\Delta_{\{c\},S} \preceq \Delta_{\{c'\},S'}$
when $\dim (\Delta_{\{c\},S}) \leq \dim(\Delta_{\{c'\},S'})$. 

Let $\cY_0 = \cX^n$ and let $\cY_k$ be the blow up of $\cY_{k-1}$ along  the iterated dominant
transform of those subvarieties $\Delta_{\{c\},S}$ whose dimension is $k$.  We simply
set $\cY_k = \cY_{k+1}$ if there are no subspaces  $\Delta_{\{c\},S}$  of dimension $k$. The 
resulting sequence of blow ups
\begin{equation}\label{Yks}
\cY_m \longrightarrow \cdots \longrightarrow \cY_1 \longrightarrow \cY_0 = \cX^n
\end{equation}
terminates when  it reaches the maximal dimension of the subvarieties $\Delta_{\{c\},S}$, that is,  
$$m= \dim(\cX) \cdot (n-1) + \max_c (\dim (\cH_c)).$$ Each step of this iteration procedure 
is independent  on the order in which the blowups are performed along the (iterated) 
dominant transforms of the subspaces $\Delta_{\{c\},S}$.  Thus, the  intermediate varieties 
$\cY_k$ in the sequence \eqref{Yks} are all well defined. 

The variety 
$$
 \cX_\cH^{[n]} := \cY_{m}
$$ 
obtained through this sequence of iterated blowups is  the 
wonderful compactification of the arrangement $\cS_n$. It is a
compactification of the configuration space of $n$ labelled 
(not necessarily distinct) points in $\cX$ 
that keeps the  points  away from the subvariety $\cH$.

\medskip
\subsubsection{Second step: Configurations of distinct points}
Let $\tilde{\Delta}_I \subset \cX_\cH^{[n]}$ denote the proper transform of $\Delta_I$.

\begin{lem}[Kim \& Sato \cite{KS},  Proposition 4]
The collection of all subspaces  $\{ \tilde{\Delta}_I  \mid I \subset \underline{n}, |I|>1\}$
is a building set of a subspace arrangement in  $\cX_\cH^{[n]}$. 
\end{lem}

Let $\cM_0 = \cX_\cH^{[n]}$ and let $\cM_{k+1}$ be the blow up of $\cM_{k}$ along  
the iterated dominant transform of $\tilde{\Delta}_{I}$ for $|I|-k$.  The resulting 
sequence of blow ups
\begin{equation}\label{Mks}
\cM_{n-1} \longrightarrow \cdots \longrightarrow \cM_1 \longrightarrow \cM_0 = \cX_\cH^{[n]}
\end{equation}
does not depend on the order in which the blowups are performed along the (iterated) 
dominant transforms of the subspaces $\tilde\Delta_{I}$, hence the 
intermediate varieties $\cM_k$ in the sequence \eqref{Mks} are 
all well defined.
The variety 
$$
\cX_\cH[n] : = \cM_{n-1}
$$ 
is a compactification of the configuration space of $n$ labelled 
points in $\cX$ that keeps the points
pairwise disjoint and away from the subvariety $\cH$.  Moreover, we 
have the following result.

\begin{thm}[Kim \& Sato \cite{KS},  Theorem 2] 
\label{Thm_Conf}
The wonderful compactifications obtained as above have the following properties:
\begin{enumerate}
\item The variety $\cX_\cH[n]$ is nonsingular.
\item The boundary 
$$
\cX_\cH[n] \setminus \left( (\cX \setminus \cH)^n \setminus  
( \bigcup_{ \{i,j\} \subset \underline{n}} \Delta_{\{i,j\}}) \right)
$$
is the union of exceptional divisors $D_{c,S}$ and $D_I$, corresponding to the subvarieties 
$\Delta_{c,S}$ for $c \in \underline{k}^+$, $|S|>0$, and the diagonals $\Delta_{I}$ for $|I| >1$. 
Any set of these divisors intersects transversally.
\item The intersection of boundary divisors $D_{c_1,S_1},\cdots,D_{c_a,S_a}, D_{I_1},\cdots,D_{I_b}$
is nonempty if and only if they are nested.
\end{enumerate}
\end{thm}

\section{Feynman integrals of Euclidean scalar QFTs}\label{FeySec}

In this section, we  describe Feynman amplitudes and unrenormalized Feynman integral
on the configuration space $\cF(\cX \setminus \cH,n)$.

\subsection{Feynman graphs} \label{graphSec1}

A {\em graph} $\Gamma$ is a  one-dimensional finite CW-complex. We denote 
the set of vertices of $\Gamma$ by $V_\Gamma$, the set of edges by 
$E_\Gamma$, and the boundary map by $\partial_{\Gamma}: E_{\Gamma} 
\to (V_{\Gamma})^2$. We write $\partial_\Gamma(e)=\{ s(e), t(e) \}$, the source 
and target vertices of the oriented edge $e$.  We denote the valency of a
vertex $v$ by $|v|$.

The set of {\it external vertices} $V^e_\Gamma$ is $\{v \in V_\Gamma\ \mid |v|=1\}$,
and the set of {\it internal vertices} $V^i_\Gamma$ is its complement $V_\Gamma
\setminus  V^e_\Gamma$.
The set of edges $E_\Gamma$ is the union of the set $E_\Gamma^e$ of 
{\it external edges},  that are the edges having   external vertices  as boundary vertices,   
and $E_\Gamma^i := E_\Gamma \setminus E_\Gamma^e$ is the set of 
{\it internal edges}.

A {\it  looping}  edge is an edge for which the two boundary vertices coincide. 
Since  the  position space Feynman amplitudes  would diverge everywhere in such 
cases, we exclude them and assume that  our graphs have no  looping edges.

\smallskip
\subsection{Feynman rules} \label{graphSec}

As we discussed already in \cite{CeyMar2}, we consider two different approaches to
defining the Feynman integrals in configuration space, which we refer to here as
the {\em real case} and the {\em complex case}. In the real case, one can choose
the manifold $\cX$ to be $\P^D(\C)$ and the physical $D$-dimensional spacetime 
manifold will be the real locus $\A^D(\R)$ inside an affine space 
$\A^D(\C) \subset \P^D(\C)$. In the {\em complex case} spacetime is complexified.
The manifold $\cX$ consists of a product $\P^D(\C)\times \P^D(\C)$ and the
spacetime is the middle dimensional subspace given by an affine space 
$\A^D(\C) \times \A^0$ lying in one of the two copies of $\P^D(\C)$. The two choices lead to
completely different methods of dealing with the resulting Feynman integrals,
due to the fact that, in the real case, the differential form describing the Feynman
amplitude is not a closed form, while in the complex case it is. We recall below
the two cases and the resulting Feynman amplitudes. 

\subsubsection{Feynman rules: real case} \label{SecRulesReal}
Let $\cX=\P^D(\C)$ and consider the real locus $\A^D(\R)$ in the affine part $\A^D(\C)$ 
as the physical space-time manifold. We also denote the hyperplane at infinity  
$\P^D \setminus \A^D$ by $\cH_\infty$.  
 
Given a Feynman graph $\Gamma$ with  $n$ internal and $k$ external vertices and no
looping edges, the Feynman rules assigns to $\Gamma$ a Feynman amplitude  as follows:

\begin{itemize}
\item The external vertices are labelled by fixed pairwise disjoint points 
$\cH_c = \{z_c\} \subset \cX$ with  $c \in \underline{k}$, which we refer to as
``external points";

\item  To each internal vertex $v_j$ one assigns a variable $z_j \in \cX$, so that
internal vertices are labelled by the set of points $(z_1,\dots, z_n) \in \cX^n$;

\item To each edge $e \in E_\Gamma$ with $\partial_\Gamma(e) = \{s(e),t(e)\}$, we 
assign a Euclidean {\it propagator}. This is given in the {\em massless} case by
\begin{equation}\label{RGm0}
G_{0,\R}(x_{s(e)}-x_{t(e)})= \frac{1}{\| x_{s(e)} - x_{t(e)} \|^{2\lambda}}, 
\ \ \  \text{ where } \ \ D=2\lambda+2.
\end{equation}
In the massive case the propagator has the form 
\begin{equation}\label{RGm}
G_{m,\R}(x_{s(e)}-x_{t(e)})= (2\pi)^{-(\lambda+1)} \, m^{\lambda}\, \|x_{s(e)}-x_{t(e)}\|^{-\lambda}\, 
\cK_{\lambda}(m \| x_{s(e)}-x_{t(e)}  \|),
\end{equation}
where $\cK_\nu(z)$ is the modified Bessel
function (Macdonald function, see \S 1.2 of \cite{AdHed}).
\end{itemize}

We will explain more in detail the origin of the propagators \eqref{RGm0}
and \eqref{RGm} in terms of potential theory in \S \ref{PotSec} below, using
the fact that the propagator \eqref{RGm0} is obtained by considering the
Green function of the {\em real} Laplacian on $\A^D(\R)$, which for
$D=2\lambda+2$ is given by
\begin{equation}\label{GreenReal}
  G_{\R,\Delta}(x-y) = \frac{1}{\| x-y \|^{2\lambda}}. 
\end{equation}  

\smallskip

The resulting Feynman amplitude associated to the graph $\Gamma$ in 
the massless real case is then the differential form
\begin{equation}\label{Romegam0}
\omega_{\Gamma,\R} = \prod_{e\in E_\Gamma} 
\frac{1}{\| x_{s(e)} - x_{t(e)} \|^{2\lambda}} 
\bigwedge_{v\in V_\Gamma} dx_v.
\end{equation}
This is a smooth differential form in degree equal to the middle dimension,
on the configuration space, $\omega_{\Gamma,\R} \in 
\Omega^{D n}(\cF(\cX\smallsetminus \cH,n))$, where we set 
$$
\cH  := \bigcup_{c \in \underline{k}^+} \cH_c.
$$
However, 
it is {\em not} a closed form, see Proposition 3.3 of \cite{CeyMar2}. Similarly,
the Feynman amplitude associated to the graph $\Gamma$ in 
the massive real case is
\begin{equation}\label{Romegam}
\omega_{\Gamma,\underline{m},\R} 
=(2\pi)^{-(\lambda+1)| E_\Gamma|}\, \prod_{e\in E_\Gamma} m_e^\lambda\,
 \|x_{s(e)}-x_{t(e)} \|^{-\lambda}\, \cK_{\lambda}(m_e \| x_{s(e)} - x_{t(e)} \|) \, \bigwedge_{v\in V_\Gamma} dx_v,
\end{equation}
where we allow the different edges to have different mass parameters, with
$\underline{m}=(m_e)_{e\in E_\Gamma}$.

\smallskip
\subsubsection{Chain of integration: real case}

In this real case, the chain of integration for the Feynman integral is given by
the middle dimensional cycle
\begin{equation}\label{chainR}
 \sigma_{\Gamma,\R} = \cX({\mathbb R})^{V_\Gamma^i} =\cX(\R)^n.
\end{equation}

\subsubsection{Feynman rules: complex case} \label{SecRulesComplex}

Let $\cX = \P^D(\C) \times \P^D(\C)$ with projection $p: \cX \to \P^D(\C)$,
$p: z=(x,y) \mapsto x$. The {\em complex} affine space $\A^D(\C)$ in
$p(\cX)$ is identified as the physical space-time in this case.

\begin{itemize}
\item The external vertices are labelled by fixed pairwise disjoint points 
$\cH_c = \{z_c=(x_c,y)\} \subset \cX$, with  $c \in \underline{k}$, and with $y$
a chosen base point in $\P^D(\C)$. 

\item  To each internal vertex $v_j$ one assigns a variable $x_j=p(z_j)$, with $z_j \in \cX$, so that
internal vertices are labelled by the set of points $(x_1,\dots, x_n) \in \P^D(\C)^n$;

\item To each edge $e\in E_\Gamma$, with $\partial(e)=\{ s(e), t(e) \}$, 
we consider the Euclidean {\em massless} edge propagator 
\begin{equation}\label{CGm0}
G_{0,\C} (x_{s(e)}-x_{t(e)})=\frac{- (D-2)!}{(2\pi i)^D} \frac{1}{\| x_{s(e)}- x_{t(e)} \|^{2D-2}},
\end{equation}
where $\| x_{s(e)} - x_{t(e)} \|=\| p(z)_{s(e)}- p(z)_{t(e)} \|$.
\end{itemize}

The resulting {\em massless} Feynman amplitude for a graph $\Gamma$ is then of the form
\begin{equation}\label{Comegam0}
\omega_{\Gamma,\C} = \left( \frac{- (D-2)!}{(2\pi i)^D} \right)^{E_\Gamma}
\prod_{e\in E_\Gamma} \frac{1}{\| x_{s(e)}- x_{t(e)} \|^{2D-2}} 
\bigwedge_{v \in V_\Gamma} dx_v \wedge d\bar x_v  .
\end{equation}
This is a {\em closed} differential form, of degree equal to the middle
dimension, in the configuration space
$\omega_{\Gamma,\C}\in \Omega^{2D n}(\cF(\cX \setminus \cH,n))$,
with $\cH =\cup_{c \in \underline{k}^+} \cH_c$.

\smallskip

In the {\em massive} case, with the complexified spacetime $\A^D(\C)$, the
edge propagator \eqref{CGm0} is replaced by the massive edge propagator 
\begin{equation}\label{CGm}
G_{m,\C} (x_{s(e)}-x_{t(e)})=(2\pi)^{-D}  m^{D-1}\, 
\|x_{s(e)}-x_{t(e)} \|^{-(D-1)}\, \cK_{D-1}(m \| x_{s(e)} - x_{t(e)} \|), 
\end{equation}
with $\cK_\nu(z)$ the modified Bessel function.

\smallskip

Thus, in the complexified case, the massive Feynman amplitude is given by
\begin{equation}\label{omegamZ}
\omega_{\Gamma,\underline{m},\C} =(2\pi)^{-D | E_\Gamma|}\, \prod_{e\in E_\Gamma} m_e^{D-1}\, 
 \|x_{s(e)}-x_{t(e)} \|^{-(D-1)}\, \cK_{D-1}(m \| x_{s(e)} - x_{t(e)} \|) \, 
 \bigwedge_{v\in V_\Gamma} dx_v \wedge d\bar x_v.
\end{equation}

\smallskip

The choice of propagators in this case arises from the fact that
on ${\mathbb A}^D({\mathbb C})$ the {\em complex} Laplacian 
$$ \Delta = \sum_{k=1}^D \frac{\partial^2}{\partial x_k \partial \bar{x}_k}  $$
has Green form
\begin{equation}\label{GreenComplex}
 G_{\C,\Delta}(x-y) = \frac{- (D-2)!}{(2\pi i)^D \| x-y \|^{2D-2}} .
\end{equation}
We will return to this more in detail in \S \ref{PotSec} below.

\smallskip
\subsubsection{Chain of integration: complex case}

In the complex case, the chain of integration for the Feynman integral is given by
the middle dimensional cycle
\begin{equation}\label{chainR}
 \sigma_{\Gamma,\C} = \P^D(\C)^{V_\Gamma^i} \times \{ y \}=\P^D(\C)^n\times \{ y \} ,
\end{equation} 
for a choice of a point $y= (y_v \mid v \in V_\Gamma^i)$.

\smallskip
\subsubsection{Unrenormalized Feynman integral}

The unrenormalized massless Feynman integral for a Feynman graph $\Gamma$ is then given by
\begin{equation}\label{WeightZ}
\begin{array}{ll}
\int_{\sigma_{\Gamma,\R}} \omega_{\Gamma,\R} & \text{in the real case} \\[3mm]
\int_{\sigma_{\Gamma,\C}} \omega_{\Gamma,\C} & \text{in the complex case,} 
\end{array}
\end{equation}
and similarly for the massive Feynman integrals.

\section{Potential theory and propagators in configuration and momentum space}\label{PotSec}
In this section we justify the form \eqref{RGm0} and \eqref{RGm} of the propagators 
in terms of potential theory, and similarly for their complex counterparts.
In particular, we show how the expressions \eqref{RGm0} and \eqref{RGm} 
arise naturally from the theory of Riesz and Bessel potentials, which gives 
a rigorous formulation of the Fourier transform relation between momentum and
configuration space propagators. 

\smallskip
\subsection{Massless propagator and the Riesz potential}

We first recall the relation between the massless Euclidean configuration 
space propagator 
\begin{equation}\label{G0}
G_{0,\R,D}(x_s - x_t) = || x_s - x_t ||^{2 -D} 
\end{equation}
for $D>2$ and the momentum space propagator, through
the relation of \eqref{G0} to the Riesz potential.

\begin{lem}\label{G0FourierLem}
The propagator \eqref{G0} satisfies the Fourier transform formula
\begin{equation}\label{G0Fourier}
\widehat{(G_{0,\R,D} \star \varphi)}(k) =\frac{4 \pi^{D/2}}{\Gamma(\lambda)}\, \frac{1}{\| k \|^2} \, 
\widehat{\varphi}(k),
\end{equation}
for any test function $\varphi \in \cS(\R^D)$.
\end{lem}

\begin{proof}
The function \eqref{G0} is the Green function of the real Laplacian on $\R^D$.
It acts on test functions by convolution product
\begin{equation}\label{G0phi}
G_{0,\R,D}\star \varphi (x) = \int_{\R^D} \frac{\varphi(y)}{\| x-y \|^{D-2}} \, dy.
\end{equation}
This is a special case of the Riesz potential (see \S 1.2 of \cite{AdHed}) 
$\cI_\alpha=C_\alpha \| x \|^{\alpha-D}$ with
$$C_\alpha = \frac{\Gamma((D-\alpha)/2)}{ 2^\alpha \pi^{D/2} \Gamma(\alpha/2)},$$
in the case where $\alpha=2$,
$$ G_{0,\R,D} (x-y)  =\frac{4 \pi^{D/2}}{\Gamma(\lambda)} \,  \cI_2(x-y) $$
where $D=2\lambda+2$. The Riesz potential satisfies the Fourier transform
property
\begin{equation}\label{IalphaFourier}
\widehat{(\cI_\alpha \star \varphi)}(k) = \frac{1}{\| k \|^\alpha} \, \widehat{\varphi}(k),
\end{equation}
hence, $G_{0,\R,D}(x-y)$ satisfies the Fourier transform formula \eqref{G0Fourier}.
\end{proof}

Notice that $\| k \|^{-2}$ is the massless momentum space propagator associated to
a single edge carrying a momentum variable $k\in \R^D$.
Thus configuration space propagator \eqref{G0} can be identified with the Fourier transform of the
momentum space propagator associated to the same edge of the Feynman graph as
\begin{equation}\label{hatG0}
\widehat{G_{0,\R,D}}(k) = \frac{4 \pi^{D/2}}{\Gamma(\lambda)}\, \frac{1}{\| k \|^2}.
\end{equation}

\smallskip
\subsubsection{Complex case}

In the complex case, the propagator \eqref{CGm0} can be similarly related to
the momentum space propagator, by interpreting it as the
Green function \eqref{GreenComplex} of the complex Laplacian on $\A^D(\C)$ 
as $G_{\C,\Delta}(x,y)=(-(D-2)!/(2\pi i)^D) \, G_{0,2D}$ and using again \eqref{G0phi} to
relate it to the massless momentum space propagator, so that we obtain
$$ G_{0,\C,D}(x_{s(e)}-x_{t(e)})= \frac{- (D-2)!}{(2\pi i)^D}\frac{4\pi^D}{\Gamma(D-1)} \cI_2 (x_{s(e)}-x_{t(e)}) = (2i)^{2-D} \cI_2 (x_{s(e)}-x_{t(e)}), $$
which satisfies
$$ \widehat{G_{0,\C,D}}(k) = (2i)^{2-D} \frac{1}{\| k \|^2}. $$

\subsection{Massive propagator and Helmholtz equation}

The case of the massive propagator can be similarly related via Fourier transform
to the Bessel potential. Notice that we use here a different sign convention for 
the Laplacian with respect to \cite{AdHed}. Let $G_{m,\R, D}(x-y)$ denote
the massive propagator of \eqref{RGm} in dimension $D=2\lambda+2$.

\begin{lem} 
The massive Euclidean configuration space propagator $G_{m,\R,D}(x-y)$ satisfies 
the Fourier transform formula
\begin{equation}\label{GmFourier}
\widehat{(G_{m,\R, D}\star \varphi)}(k) = \frac{1}{(m^2 + \| k \|^2)} \, \widehat{\varphi}(k)
\end{equation}
for any test function $\varphi \in \cS(\R^D)$. In particular, $G_{m,\R,D}$ is the
Green function of the Helmholtz equation with wavenumber $m^2$.
\end{lem}

\begin{proof} 
The fractional integral operator $(m^2 \cdot I+\Delta)^{-\alpha/2}$, where $m\neq 0$ is 
the mass parameter, satisfies the Fourier transform property
\begin{equation}\label{BesselFourier}
\widehat{((m^2\cdot I+\Delta)^{-\alpha/2} \varphi)}(k) = \frac{1}{(m^2 + \| k \|^2)^{\alpha/2}} \, \widehat{\varphi}(k),
\end{equation}
hence it can be represented by the convolution operator
\begin{equation}\label{BesselGm}
(m^2 \cdot I+\Delta)^{-\alpha/2} \varphi = G_{m,\R,D}^{(\alpha)} \star \varphi, \ \  \text{ with } \ \
\widehat G_{m,\R,D}^{(\alpha)}(k) = \frac{1}{(m^2 + \| k \|^2)^{\alpha/2}}.
\end{equation}
In the case $m=1$, this is the usual Bessel potential $\cG_{\alpha,\R,D}$ 
(see \S 1.2 of \cite{AdHed}) and, in particular,  
when $\alpha=2$, the Green function $G_{m,\R,D}^{(2)}$ is the fundamental solution
of the Helmholtz equation with wavenumber $m^2$, namely it satisfies
$$ 
\Delta G^{(2)}_{m,\R,D} + m^2 G^{(2)}_{m,\R,D}  = \delta.
$$
In the case $\alpha=2$, the Green function $G_{m,\R,D}^{(2)}$ gives the massive Euclidean
configuration space propagator $G_{m,D}$, since
the momentum space Euclidean massive propagator is of the form
$$ 
\frac{1}{m^2 + \| k \|^2}, 
$$
with $k\in \R^D$ the momentum variable associated to the edge.
\end{proof}

\smallskip
\subsubsection{Complex case}

In the complex case, similarly, the massive propagator is a solution of
the Helmholtz equation with the {\em complex} Laplacian and wavenumber $m^2$,
$$ \Delta_\C G_{m,\C,D} + m^2  G_{m,\C,D} =\delta. $$
It will be clear from \S \ref{BesselSec} below that this gives the form of the
propagator given in \eqref{CGm}.

\medskip
\subsection{Massive propagator and the Bessel potential}\label{BesselSec}
The massive propagator $G_{m,D}$ can then be described explicitly 
in terms of Bessel functions as follows. 

\begin{lem}\label{GmIntlem}
The massive Euclidean propagator $G_{m,\R,D}$ on $\R^D$ has the following integral representation
\begin{equation}\label{GmInt}
G_{m,\R,D}(x)= \frac{1}{(4 \pi)^\frac{D}{2}}  \int_0^\infty dt  \  \ \ t^{-\frac{D}{2}} \ e^{- t m^2 - \frac{ \| x \|^2}{4 t} }.   
\end{equation}
\end{lem}

\begin{proof}
By using the following integral form of the momentum space  propagator
$$
\frac{1}{\| k \|^2 +m^2 } = \int_0^\infty dt \ e^{ - t  \cdot (\| k \|^2 +m^2) }
$$
we can give its Fourier transfom \eqref{GmFourier} as follows
\begin{eqnarray*}
G_{m,\R,D} (x) 
& = &   \frac{1}{(2 \pi)^D} \int_{\R^D} dk  \ e^{i k \cdot x} \ \int_0^\infty dt \ e^{ - t  \cdot (\| k \|^2 +m^2) } \\
& = &   \frac{1}{(2 \pi)^D}  \int_0^\infty dt \ e^{- t m^2}   \int_{\R^D} dk  \  e^{ - t   \| k \|^2 +i k \cdot x + \frac{ \| x \|^2}{4 t}}  
\ \cdot \ e^{ - \frac{ \| x \|^2}{4 t} }   \\
& = &   \frac{1}{(2 \pi)^D}  \int_0^\infty dt \ e^{- t m^2 - \frac{ \| x \|^2}{4 t} }   \int_{\R^D} dq  \  e^{ - t   \| q \|^2 } 
\ \  \text{where} \ q = k - \frac{ i x}{2 t}  
\end{eqnarray*}
and obtain \eqref{GmInt} by substituting $\left( \frac{2 \pi }{2 t} \right)^{\frac{D}{2}} =  \int_{\R^D} dq  \  e^{ - t   \| q \|^2 }$.
\end{proof}

In particular, this implies the following behavior of divergences.

\begin{cor} \label{CorDiag}
For $x=0$, we have
$$
G_{m,\R,D} (0) = \frac{1}{(4 \pi)^\frac{D}{2}}  \int_0^\infty dt  \  \ \ t^{-\frac{D}{2}} \ e^{- t m^2 } 
= \frac{1}{(4 \pi)^\frac{D}{2}}  \  m^{D-2} \  \Gamma (1 - \frac{D}{2}),
$$
that diverges when $D= 2 \lambda + 2$ for $\lambda \in \N^+$.
\end{cor}

It is convenient, in order to write $G_{m,\R,D}$ in terms of Bessel functions, to first
encode the mass dependence in a scaling relation in terms of the propagator at $m=1$.

\begin{lem}\label{massGm}
The dependence of the Green function $G_{m,D,\alpha}$ 
of the operator $\cG_{m,D,\alpha}=(mI +\Delta)^{-\alpha/2}$ on
the mass parameter $m$ is given by
\begin{equation}\label{massG2}
G_{m,D,\alpha}(x)= m^{D-\alpha} \, G_{D,\alpha} (mx),   
\end{equation}
where $G_{D,\alpha}(x)=G^{(\alpha)}_{m=1}$ is the Bessel potential on $\R^D$.
\end{lem}

\proof This follows from the integral representation \eqref{GmInt} for $G_{m,D,\alpha}(x)$
and the integral representation 
\begin{equation}\label{eqnBes1}
G_{1,D}^{(\alpha)}(x) = a_\alpha \, \int_0^\infty t^{\frac{\alpha-D}{2}}\, e^{- \frac{\pi \| x\|^2}{t} - \frac{t}{4\pi}}\, \frac{dt}{t}, 
\ \ \text{where} \ \ a_\alpha = \frac{1}{(4 \pi)^{\alpha /2} \Gamma (\alpha/2)}
\end{equation}
for the Bessel potential (see \S 1.2.4 of \cite{AdHed}). Using the change of
variables $t  =  4 \pi m^2 \cdot t' $ in \eqref{eqnBes1} gives \eqref{massG2}.
\endproof

The scaling \eqref{massG2} can also be seen in the Fourier representation
$$ G_{D,\alpha}(x) = \frac{1}{(2\pi)^D} \int_{\R^D} 
\frac{e^{i x\cdot \xi}}{ (1+ \| \xi \|^2)^{\alpha/2} } \, d\xi, $$
where, after changing variables to $k =\xi /m$, one gets
$$ G_{m,D,\alpha}(x) = \frac{1}{(2\pi)^D} \int_{\R^D} \frac{m^D\, e^{i m \, x\cdot \eta}}{m^\alpha (1+ \| \eta \|^2)^{\alpha/2} } \, d\eta, $$
which gives \eqref{massG2}. Since both $\widehat{G_{D,\alpha}}$ and $\widehat{G_{m, D,\alpha}}$ are entire analytic functions of $\alpha$ for each fixed $\xi$ or $\eta$, the corresponding
distributions $G_\alpha$ and $G_{m,\alpha}$ have the property that 
$\langle G_{D,\alpha},\varphi \rangle$ and $\langle G_{m,D,\alpha},\varphi \rangle$ 
are analytic in $\alpha$ for all test functions $\varphi\in \cS$,
and the identity \eqref{massG2} extends for all $\alpha$.

\smallskip

We then obtain the following expression for the massive propagator 
in terms of Bessel functions.

\begin{lem}\label{GmBessellem}
The massive Euclidean propagator $G_{m,D}$ on $\R^D$ satisfies
\begin{equation}\label{GmBessel}
G_{m,\R,D}(x-y)= (2 \pi)^{ -\frac{D}{2}} \ m^{D-2} \ \left(m  \|  x - y \| \right)^{ - \frac{D-2}{2}} \ \cK_{\frac{D-2}{2}} (m \| x-y \|),
\end{equation}
where $\cK_\nu(z)$ is the modified Bessel
function (Macdonald function).
\end{lem}

\begin{proof} This can be seen from the integral representation of Lemma \ref{GmIntlem},
or equivalently from Lemma \ref{GmBessellem}. We use again the fact that 
the Bessel potential satisfies \eqref{eqnBes1},
which implies (\cite{Titch} p.~201 and \S 1.2 of \cite{AdHed}) that it can be expressed 
in terms of Bessel functions as
\begin{equation}\label{eqnBes2}
G_{1,D}^{(\alpha)}(x) =  c_\alpha \, \| x\|^{-\frac{D-\alpha}{2}} \, \cK_{\frac{D-\alpha}{2}}(\| x \|), 
\end{equation}
with $c_\alpha^{-1} = 2^{(\alpha-2)/2} \, (2\pi)^{D/2} \Gamma(\alpha/2)$, and 
with $\cK_\nu(z)$ the modified Bessel function. 

The statement follows by  applying again the change of variables $t  =  4 \pi m^2 \cdot t' $ 
in the integral \eqref{eqnBes1} and using the identities 
\eqref{eqnBes2} and \eqref{GmInt}  for case of $\alpha=2$. Equivalently,
\eqref{GmBessel} follows from the scaling property 
of Lemma \ref{massGm} and \eqref{eqnBes2}.
\end{proof}

\smallskip

The integral representation of Lemma \ref{GmIntlem} also gives 
an estimate on the behavior of the massive propagator at infinity.

\begin{cor} \label{CorInf}
The massive propagator $G_{m,\R,D}(x)$ is smooth on $\R^D \setminus \{0\}$. Moreover, 
$G_{m,\R,D}(x) \to 0$ as $\|x\| \to \infty$. 
\end{cor}

\begin{proof}
The smoothness of $G_{m,D}(x)$   is a direct consequence of the fact 
that the modified Bessel function $\cK_\nu(z)$ is  analytic on $\C^*$ and has an essential 
singularity at $z = \infty$. 

The Cauchy--Schwarz inequality takes the from
\begin{eqnarray*} 
| G_{m,\R,D}(x)|^2  &\leq& 
\frac{1}{(4 \pi)^D} 
\left( \int_0^\infty dt  \  \  e^{- 2 t m^2 } \right)
\cdot
\left(   \int_0^\infty dt  \  \ \ t^{-D} \ e^{- \frac{ \| x \|^2}{2 t} } \right) 
= \frac{1}{(4 \pi)^D} \frac{1}{2 m^2}  
\int_0^\infty dt  \  t^{-D} \ e^{- \frac{ \| x \|^2}{2 t} } \\
&=& \frac{2^{D-1}}{(4 \pi)^D} \frac{ 1}{2 m^2}  \|x\|^{-2D +2}
\int_0^\infty \frac{dt'}{t'}  e^{-t'} t'^{D-1} 
= \frac{2^{D-1}}{(4 \pi)^D} \frac{1}{2 m^2}  \Gamma (D-1) \|x\|^{-2D +2} \\
&= & \frac{2^{D-1}}{(4 \pi)^D} \frac{(D-2)!}{2 m^2}   \|x\|^{-2D +2}
\end{eqnarray*} 
after the coordinate transformation $t' = \frac{\| x \|^2}{ 2 t}$, and it simply
implies that $G_{m,\R,D}(x)$ in \eqref{GmInt} tends to zero at infinity for $D>1$ 
due to the term $\|x\|^{-2D +2}$. 

This can also be seen using the integral representation \eqref{eqnBes1} for $G_{1,D}^{(\alpha)}$
for $\alpha=2$ and the integral representation of the Riesz potential (see \S 1.2.4 of \cite{AdHed})
$$ I_\alpha(x) = a_\alpha \int_0^\infty t^{\frac{\alpha-D}{2}}e^{-\frac{\pi \| x \|^2}{t}} \, \frac{dt}{t}, $$
which then gives an estimate (for all $0<\alpha<D$)
$$ 0\leq G_{1,D}^{(\alpha)}(x) \leq  I_\alpha(x). $$
\end{proof}

\smallskip
\subsubsection{Complex case}

The same arguments used in Lemma \eqref{GmBessellem} and Corollary \ref{CorInf}
give an expression of the massive propagator in the {\em complex case} in terms
of modified Bessel functions.

\begin{lem}\label{CGmBessel}
The massive Euclidean propagator $G_{m,\C,D}$ satisfies
\begin{equation}\label{GmCBessel} 
G_{m,\C,D} (x-y)=(2\pi)^{-D}  m^{D-1}\, 
\|x -y \|^{-(D-1)}\, \cK_{D-1}(m \| x - y \|)
\end{equation}
In particular $G_{m,\C,D}$ is smooth away from the origin, with
$G_{m,\C,D}(x) \to 0$ as $\|x\| \to \infty$.
\end{lem}

\subsection{Expansions of massive Feynman amplitudes}

Using the series expansions of the Bessel functions, we obtain the following formal expansion for the massive Euclidean 
propagator in terms of the distance function.

\begin{lem} \label{TaylorProp}
Let $D= 2 \lambda+2$ , then the massive Euclidean propagator has an  expansion of the form
\begin{eqnarray*}
G_{m,\R,D}(x-y) &=& \left(\frac{1}{2 \pi} \right)^{(\lambda+1)}  
\sum_{\ell =0}^{\lambda-1} \frac{(-m^2)^\ell \cdot 2^{\lambda - 2 \ell -1} \cdot (\lambda - \ell -1)!}{\ell! \   \| x- y \|^{2(\lambda -  \ell)}} \\
&+& \left(\frac{-1}{2 \pi} \right)^{(\lambda+1)}   \sum_{\ell = 0}^\infty 
\frac{ m^{2(\lambda + \ell)} \| x- y \|^{ 2 \ell}}{2^{\lambda + 2 \ell} \cdot \ell! \cdot  (\lambda + \ell)! } 
 \left( 
 \log (\frac{ m \| x- y \|}{2})) - \frac{ \psi (\ell +1)   +  
  \psi (\lambda + \ell +1)}{2} 
 \right)
\end{eqnarray*}
where $\psi(x)  = \frac{d}{dx} \ln \Gamma(x) = \frac{\Gamma'(x)}{\Gamma(x)}$ is the digamma function.
\end{lem}

\begin{proof}
We use the following expansion of modified Bessel function $\cK_\lambda(z)$, see p. 80 in \cite{Watson},
\begin{eqnarray*}
\cK_\lambda(z) = \frac{1}{2} \sum_{\ell =0}^{\lambda-1} \frac{(-1)^\ell (\lambda - \ell -1)!}{\ell! \ (\frac{z}{2} )^{\lambda - 2 \ell}} 
 + (-1)^{\lambda+1} \sum_{\ell = 0}^\infty 
 \frac{(\frac{z}{2} )^{\lambda + 2 \ell}}{\ell! \ (\lambda + \ell)! } 
 \left( 
 \log (\frac{z}{2})) - \frac{ \psi (\ell +1)  +  
 \psi (\lambda + \ell +1)}{2} 
 \right).
\end{eqnarray*}
in \eqref{GmBessel}.
\end{proof}

The complex case is analogous. One obtains the following expression.

\begin{cor}\label{CTaylorCor}
In the complex case the massive Euclidean propagator has an expansion
\begin{eqnarray*}
G_{m,\C,D}(x-y) &=& \frac{1}{(2 \pi)^D} 
\sum_{\ell =0}^{D-2} \frac{(-m^2)^\ell \cdot 2^{D-2 - 2 \ell} \cdot (D-2 - \ell)!}{\ell! \  
 \| x- y \|^{2(D-1 -  \ell)}} \\
&+& \left(\frac{-1}{2 \pi} \right)^D   \sum_{\ell = 0}^\infty 
\frac{ m^{2(D-1 + \ell)} \| x- y \|^{ 2 \ell}}{2^{D-1 + 2 \ell} \cdot \ell! \cdot  (D-1 + \ell)! } 
 \left( 
 \log (\frac{ m \| x- y \|}{2})) - \frac{ \psi (\ell+1)   +  
 \psi (D + \ell)}{2} 
 \right)
\end{eqnarray*}
\end{cor}

This gives a corresponding expansion for the Feynman amplitudes.

\begin{prop}\label{Taylorexpand}
The Feynman amplitude $\omega_{\Gamma, \mathbf{m}, \R} $ of \eqref{Romegam} has
an  expansion
\begin{equation}\label{tayloromega}
\omega_{\Gamma,\mathbf{m}, \R} = \sum_{\underline{\ell}=(\ell_e)\in \{-\lambda,-\lambda+1, \cdots,\infty \}^{E_\Gamma}}
\left(\frac{1}{2 \pi} \right)^{(\lambda+1) |E_\Gamma|}   \prod_{e\in E_\Gamma} 
B^\R_{\Gamma,\mathbf{m},D,\ell_e} \,\,\, \beta^\R_{\Gamma,\mathbf{m},D,\ell_e}
 \bigwedge_{v\in V_\Gamma} dx_v ,
\end{equation}
where
\begin{equation}\label{BGammaell}
B^\R_{\Gamma,\mathbf{m},D,\ell_e} := 
\left\{ 
\begin{array}{ll}
\frac{(-m_e^2)^{\lambda + \ell_e} \cdot 2^{- \lambda -  \ell_e -1} \cdot ( - \ell_e -1)!}{(\lambda + \ell_e)! }
& \ell_e \in \{-\lambda, \cdots,-1 \} \\
\frac{ (-1)^{\lambda+1} \, m_e^{2(\lambda + \ell_e)} }{2^{\lambda + 2 \ell_e} \cdot \ell! \cdot  (\lambda + \ell_e)! } 
 & \ell_e \in \{0,1,\cdots,\infty\}
 \end{array}
\right.
\end{equation}
\begin{equation}\label{formellTaylor}
\beta^\R_{\Gamma,\mathbf{m},D,\ell_e} :=  
\left\{ 
\begin{array}{ll}
\| x_{s(e)}- x_{t(e)} \|^{-2\ell_e}
& \ell_e \in \{-\lambda, \cdots,-1 \} \\
\| x_{s(e)}- x_{t(e)} \|^{ 2 \ell_e}
 \left( 
 \log (\frac{ m_e \| x_{s(e)}- x_{t(e)} \|}{2})) - \frac{ \psi (\ell_e +1)   +  
  \psi (\lambda + \ell_e +1)}{2} 
 \right) & 
  \ell_e \in \{0,1,\cdots,\infty\}.
\end{array}
\right.
\end{equation}
\end{prop}

\proof The statement   follows by  substituting the Laurent expansion of propagator in Lemma \ref{TaylorProp} 
in \eqref{omegamZ}.
\endproof

\begin{cor}
The Feynman amplitude $\omega_{\Gamma, \mathbf{m}, \C}$ of \eqref{omegamZ} 
has an expansion

\begin{equation}\label{tayloromegaC}
\omega_{\Gamma,\mathbf{m}, \C} = \sum_{\underline{\ell}=(\ell_e)\in \{ -D+1,-D+2,  \cdots,\infty \}^{E_\Gamma}}
\left(\frac{1}{2 \pi} \right)^{D |E_\Gamma|}   \prod_{e\in E_\Gamma} 
B^\C_{\Gamma,\mathbf{m},D,\ell_e} \,\,\, \beta^\C_{\Gamma,\mathbf{m},D,\ell_e}
 \bigwedge_{v\in V_\Gamma} dx_v \wedge d\bar x_v ,
\end{equation}
with $B^\C_{\Gamma,\mathbf{m},D,\ell_e}$ and $\beta^\C_{\Gamma,\mathbf{m},D,\ell_e}$
respectively as in \eqref{BGammaell} and \eqref{formellTaylor}, after replacing $\lambda$
with $D-1$.
\end{cor}

\subsubsection{Asymptotic expansions of massive propagators}

The asymptotic expansion of the Bessel functions also gives a 
formal expansion for the massive Euclidean propagator in terms of the distance
function.

\begin{lem}\label{Gmdistancelem}
The massive Euclidean propagator has an asymptotic expansion of the form
\begin{equation}\label{Gmdist}
G_{m,\R,D}(x-y) \sim \sqrt{\frac{\pi}{2}} \ \frac{e^{-m \| x-y \|}}{ ( 2 \pi)^{\lambda+1}} 
 \sum_{\ell=0}^\infty \frac{(\lambda,\ell)}{ 2^\ell m^{\ell- \lambda + 1/2}   } \| x- y \|^{-(\ell + \lambda + \frac{1}{2})},
\end{equation}
where $D=2\lambda+2$ and
\begin{equation}\label{lambdaell}
(\lambda,\ell) = \frac{\Gamma(\lambda+\ell+1/2)}{\ell! \, \Gamma(\lambda - \ell +1/2)}
\end{equation}
\end{lem}

\proof The modified Bessel function $\cK_\nu(z)$ has  asymptotic expansion, see p. 202 in \cite{Watson}
\begin{equation}\label{Knuexp}
 \cK_\nu(z) \sim \left(\frac{\pi}{2 z} \right)^{1/2}\, e^{-z} \, \sum_{\ell=0}^\infty \frac{(\nu,\ell)}{(2z)^\ell}, 
\end{equation} 
with
\begin{equation}\label{nuell}
 (\nu,\ell) = \frac{\Gamma(\nu+\ell+1/2)}{\ell! \, \Gamma(\nu - \ell +1/2)}. 
\end{equation} 
\endproof

\subsubsection{Asymptotic expansions of Feynman amplitudes}
We obtain a corresponding formal expansion for the Feynman amplitude.

\begin{prop}\label{Greenexpand}
The Feynman amplitude $\omega_{\Gamma,\mathbf{m}, \R}$ of \eqref{Romegam} has
asymptotic expansion
\begin{equation}\label{expandomega}
\omega_{\Gamma,\mathbf{m}, \R} \sim \sum_{\underline{\ell}=(\ell_e)\in \N^{E_\Gamma}}
C^\R_{\Gamma,\mathbf{m},D,\underline{\ell}} \,\,\, \omega_{\Gamma,\mathbf{m},D,\underline{\ell},\R},
\end{equation}
where
\begin{equation}\label{CGammaell}
C^\R_{\Gamma,\mathbf{m},D,\underline{\ell}} := 2^{- \frac{|E_\Gamma|}{2}} \  \pi^{-(\lambda - \frac{1}{2})|E_\Gamma|}
 \, \prod_e \frac{(\lambda,\ell_e) }{ 2^{\ell_e + \lambda +1} m_e^{\ell_e - \lambda + 1/2} }
\end{equation}
\begin{equation}\label{formellm}
\omega_{\Gamma,\mathbf{m},D,\underline{\ell},\R} :=  
\prod_{e\in E_\Gamma} 
\frac{e^{-m_e \sum_e \| x_{s(e)}-x_{t(e)} \|}} {\| x_{s(e)}- x_{t(e)} \|^{(\ell_e +\lambda + \frac{1}{2})}}  \bigwedge_{v\in V_\Gamma} dx_v .
\end{equation}
\end{prop}

\proof The statement   follows by  substituting the asymptotic expansion of propagator \eqref{Gmdist} 
in \eqref{Romegam}.
\endproof

\subsubsection{Complex case}

Similarly, we obtain an asymptotic expansion for the massive
propagator in the complex case and a corresponding expansion
of the Feynman amplitude.

\begin{cor}\label{GreenexpandC}
The massive Euclidean propagator in the complex case has an asymptotic expansion of the form
\begin{equation}\label{Gmdist}
G_{m,\C,D}(x-y) \sim 
\sqrt{\frac{\pi}{2}} \ \frac{e^{-m \| x-y \|}}{  (2\pi)^D   } 
 \sum_{\ell=0}^\infty \frac{(D-1,\ell)}{ 2^\ell m^{\ell- D+ 3/2}   } \| x- y \|^{-(\ell + D - \frac{1}{2})}.
\end{equation}
The Feynman amplitude $\omega_{\Gamma,\mathbf{m}, \C}$ of \eqref{omegamZ} has
asymptotic expansion
\begin{equation}\label{expandomega}
\omega_{\Gamma,\mathbf{m}, \C} \sim \sum_{\underline{\ell}=(\ell_e)\in \N^{E_\Gamma}}
C^\C_{\Gamma,\mathbf{m},D,\underline{\ell}} \,\,\, \omega_{\Gamma,\mathbf{m},D,\underline{\ell},\C},
\end{equation}
where the coefficients $C^\C_{\Gamma,\mathbf{m},D,\underline{\ell}}$ are as
in \eqref{CGammaell}, with $\lambda$ replaced by $D-1$, and the forms
are given by
 \begin{equation}\label{formellmC}
\omega_{\Gamma,\mathbf{m},D,\underline{\ell},\C} :=  
\prod_{e\in E_\Gamma} 
\frac{e^{-m_e \sum_e \| x_{s(e)}-x_{t(e)} \|}} {\| x_{s(e)}- x_{t(e)} \|^{(\ell_e +D- \frac{1}{2})}}  
\bigwedge_{v\in V_\Gamma} dx_v \wedge d\bar x_v.
\end{equation}
\end{cor}

\smallskip

\begin{rem}\label{Stokesrem} {\rm
The asymptotic expansion of the Bessel functions exhibits an interesting
Stokes phenomenon, see \cite{Boyd}. It would be interesting to interpret
this in the setting of Feynman integrals, perhaps along the lines of the
role of monodromies in the massive momentum space integrals in \cite{BK}.
}\end{rem}

\section{Real case: $x$-space method and Gegenbauer polynomials}\label{GegenSec}

We explain in this section the reason why it is preferable to pass from
the {\em real} to the {\em complex} formulation of Feynman amplitudes.
As we already showed in \cite{CeyMar2}, since in the real case the 
form is not closed, one cannot appeal to cohomological arguments to
reformulate the Feynman integral computation, and the only 
available method is a direct expansion of the Feynman amplitude in
orthogonal polynomials and integration on certain chains in the
configuration space that lead to nested sums. We have developed the
corresponding computation in detail in the massless case in \S 4 of
\cite{CeyMar2}, so we will focus here on how one may be able to
set up a similar computational machinery for the massive case. As one
already sees in the massless case, the computations quickly become 
intractable as the graph grows in complexity. By contrast, as we will 
see in the following section, \S \ref{RBSec}, the {\em complex} case leads directly
to a much more elegant algebro-geometric formalism for the
regularization and renormalization of Feynman amplitudes.

\medskip
\subsection{Gegenbauer polynomials and spherical harmonics}

We recall briefly some well known facts about Gegenbauer polynomials
and their relation to spherical harmonics, which we need in the following.
We refer the reader to \cite{SW} for a more detailed treatment.

\smallskip
\subsubsection{Generating function an orthogonality}

The Gegenbauer polynomials $C^{(\lambda)}_n(x)$ are determined
by the generating function, for $|t|<1$,
\begin{equation}\label{GegenGen}
\frac{1}{(1-2tx+t^2)^\lambda} = \sum_{n=0}^\infty C^{(\lambda)}_n(x)\, t^n .
\end{equation}
They satisfy the orthogonality relation, for $\lambda > 1/2$,
\begin{equation}\label{Geg2}
 \int_{-1}^1 C_n^{(\lambda)}(x) C_m^{(\lambda)}(x) \, (1-x^2)^{\lambda -1/2}  dx =\delta_{n,m}
\frac{\pi 2^{1-2\lambda} \Gamma(n+2\lambda)}{n! (n+\lambda) \Gamma(\lambda)^2}. 
\end{equation}

\smallskip
\subsubsection{Zonal and spherical harmonics}

Let $\{ Y_j \}$ be an orthonormal basis  of the Hilbert space $\cH_n(S^{D-1})$ 
of spherical harmonics on $S^{D-1}$ of degree $n$. 
The  {\em zonal spherical harmonics} $Z_{\omega_1}^{(n)}(\omega_2)$
are given by
\begin{equation}\label{zonalS2}
 Z_{\omega_1}^{(n)}(\omega_2) = \sum_{j=1}^{\dim \cH_n(S^{D-1})} Y_j(\omega_1) \overline{Y_j(\omega_2)}. 
\end{equation}

The Gegenbauer polynomials are related to 
the zonal spherical harmonics by 
\begin{equation}\label{zonalS}
C^{(\lambda)}_n(\omega_1\cdot \omega_2) = c_{D,n}\, Z_{\omega_1}^{(n)}(\omega_2),
\end{equation}
for $D=2\lambda+2$, with $\omega_1,\omega_2\in S^{D-1}$, 
with the coefficient $c_{D,n}$ given by
\begin{equation}\label{cDn}
 c_{D,n} = \frac{Vol(S^{D-1})\, (D-2)}{2n +D-2}.
\end{equation} 

\smallskip
\subsubsection{Expansions in Gegenbauer polynomials} 

A general procedure for obtaining expansions in Gegenbauer polynomials
is described in \cite{BDS}. For our purposes, it suffices to consider 
the particular case of the expansion
\begin{equation}\label{tmGegen}
x^m = 2^{-m} \Gamma(\lambda) m! \sum_{k=0}^{[m/2]} \frac{(\lambda +m - 2k)}
{k! \,\Gamma(\lambda + m +1 -k)} \, C^{(\lambda)}_{m-2k}(x).
\end{equation}
for a non-negative integer $m$ and for $\lambda \geq 1/2$.

\smallskip
\subsubsection{Chebyshev polynomials of the first kind}

The Chebyshev polynomials of the first kind have generating function 
\begin{equation}\label{ChebGenerate}
- \frac{1}{2} \log ( 1-2tx+t^2  ) = \sum_{n=1}^\infty T_n(x) \frac{t^n}{n} .
\end{equation}
They are expressible as a limit of Gegenbauer polynomials,
\begin{equation}\label{ChebGeglim}
T_n(x) =\frac{n}{2} \lim_{\lambda \to 0} \frac{C^{(\lambda)}_n(x)}{\lambda} .
\end{equation}
They can also be written as combinations of Gegenbauer polynomials of a fixed
weight $\lambda$: the expression as a function of Gegenbauer polynomials 
can be obtained inductively, using the recurrence relation $T_0(x)=1$,
$T_1(x)=x$, and $T_{n+1}(x)=2 x T_n(x)-T_{n-1}(x)$ for the 
Chebyshev polynomials, together with \eqref{tmGegen}. We denote by
$c^\lambda_{n,m}$ the resulting coefficients,
\begin{equation}\label{clamnm}
T_n(x)  =\sum_{m=0}^n  c^\lambda_{n,m} \, C^{(\lambda)}_m(x).
\end{equation}

\medskip
\subsection{Expansion of Feynman amplitudes in Gegenbauer polynomials}

We described in \S 4  of \cite{CeyMar2} how to expand the real Euclidean
massless Feynman amplitudes in Gegenbauer polynomials. Here we extend
the procedure to the massive case. We use the expansion of the
massive Feynman amplitudes given in Proposition \ref{Taylorexpand}. 

\begin{thm}\label{GegenFeyn}
The forms $(2\pi)^{-(\lambda+1)} B^\R_{\Gamma,\mathbf{m},D,\ell_e} \,
\beta^\R_{\Gamma,\mathbf{m},D,\ell_e}$ in the expansion
\eqref{tayloromega} of the real massive Euclidean Feynman amplitude 
have a series expansion in Gegenbauer polynomials, either of the form
$$ \frac{B^\R_{\Gamma,\mathbf{m},D,\ell_e}}{(2\pi)^{(\lambda+1)}}\,
\beta^\R_{\Gamma,\mathbf{m},D,\ell_e} = \rho_e^{-2\ell_e}
\sum_{n_e=0}^\infty (\frac{r_e}{\rho_e})^{n_e} \sum_{k_e=0}^{[n_e/2]}
f^\lambda_{n_e,k_e} \sum_{j_e=0}^{[(n_e-2k_e)/2]} d^\lambda_{n_e,k_e,j_e} \,
C^{(\lambda)}_{n_e-2(k_e+j_e)}( \omega_{s(e)}\cdot \omega_{t(e)}) $$
or of the form
$$ \frac{B^\R_{\Gamma,\mathbf{m},D,\ell_e}}{(2\pi)^{(\lambda+1)}}\,
\beta^\R_{\Gamma,\mathbf{m},D,\ell_e} =
\rho_e^{2\ell_e} ( \log(\frac{m_e}{2} \rho_e) \sum_{n_e=0}^{\ell_e}
a^\lambda_{n_e} (\frac{r_e}{\rho_e})^{n_e} \sum_{k_e=0}^{[n_e/2]}
b^\lambda_{n_e,k_e} \, C^{(\lambda)}_{n_e-2k_e}( \omega_{s(e)}\cdot \omega_{t(e)}) $$
$$ + \rho_e^{2\ell_e}  \sum_{n_e=0}^{\ell_e}
a^\lambda_{n_e}\sum_{k_e=0}^{[n_e/2]}
b^\lambda_{n_e,k_e} \sum_{p_e=0}^\infty \frac{1}{p_e} (\frac{r_e}{\rho_e})^{n_e+p_e}
\sum_{s_e=0}^{p_e} c^\lambda_{p_e,s_e} \sum_{r_e=0}^{[(p_e+s_e)/2]}
\alpha^\lambda_{r_e,p_e,m_e} \, \cdot $$
$$ \cdot \, \sum_{j_e=0}^{[(p_e+s_e-2r_e)/2]} \beta^\lambda_{r_e,j_e,p_e,s_e}
C^{(\lambda)}_{p_e+s_e-2(r_e+j_e)}(\omega_{s(e)}\cdot \omega_{t(e)}). $$
In the case where $D=2\lambda+2$ with $\lambda$ an integer, the 
expansion has coefficients $a^\lambda$, $b^\lambda$, $c^\lambda$,
$d^\lambda$,  $f^\lambda$,  $\alpha^\lambda$, and $\beta^\lambda$ 
in 
$$\Q[m_e, \log(m_e), \pi^{-1}, \gamma, \log(2)], 
$$
where $\gamma$ is the Euler-Mascheroni constant. 
In the case with $D=2\lambda+2$ with $\lambda$ a half-integer, the expansion
has coefficients in $$\Q[m_e, \log(m_e), \sqrt{\pi}^{\pm 1}, \gamma, \log(2)].$$
\end{thm}

\proof Consider the forms $\beta^\R_{\Gamma,\mathbf{m},D,\ell_e}$,
given as in \eqref{formellTaylor}. In the case 
$$ \beta^\R_{\Gamma,\mathbf{m},D,\ell_e} = \| x_{s(e)}-x_{t(e)} \|^{-2\ell} $$
we have an expansion
$$ \| x_{s(e)}-x_{t(e)} \|^{-2\ell} = \rho_e^{-2\ell_e} (1+(\frac{r_e}{\rho_e})^2 - 
\frac{r_e}{\rho_e} \omega_{s(e)}\cdot \omega_{t(e)} )^{-\ell_e} = $$
$$ \rho_e^{-2\ell_e} \sum_{n_e=0}^\infty C^{(\ell_e)}_{n_e}(\omega_{s(e)}\cdot \omega_{t(e)})
\, (\frac{r_e}{\rho_e})^{n_e}. $$
This is already an expansion in Gegenbauer polynomials, but with varying weight $\ell_e$.
We need instead an expansion in the Gegenbauer polynomials $C^{(\lambda)}_n$ with
the fixed $\lambda$ related to the spacetime dimension by $D=2\lambda +2$, which
are the polynomials related to the spherical harmonics on $S^{D-1}$.

We can turn the expansion above into an expansion in Gegenbauer polynomials
$C^{(\lambda)}_n$, using the explicit expression for the Gegenbauer
polynomials
\begin{equation}\label{Gegenexpl}
C^{(\lambda)}_n(x)=\sum_{k=0}^{[n/2]} (-1)^k \frac{\Gamma(\lambda + n -k)}{k! (n-2k)! 
\Gamma(\lambda)} \, (2 x)^{n-2k},
\end{equation}
together with \eqref{tmGegen}, which gives
\begin{equation}\label{GegenGegen}
C^{(\ell_e)}_{n_e}(x)=\sum_{k=0}^{[n_e/2]} (-1)^k \frac{\Gamma(\ell_e + n_e -k) \Gamma(\lambda)}
{k! \,  \Gamma(\ell_e)} \,  
\sum_{j=0}^{[(n_e-2k)/2]} \frac{(\lambda +n_e-2(k +j))}
{j! \,\Gamma(\lambda + n_e-2k +1 -j)} \, C^{(\lambda)}_{n_e-2(k+j)}(x).
\end{equation}
Thus, we obtain
$$  \| x_{s(e)}-x_{t(e)} \|^{-2\ell} = $$
$$ \rho_e^{-2\ell_e} \sum_{n_e=0}^\infty (\frac{r_e}{\rho_e})^{n_e} 
\sum_{k=0}^{[n_e/2]} (-1)^k \frac{\Gamma(\ell_e + n_e -k) \Gamma(\lambda)}
{k! \,  \Gamma(\ell_e)} \,  
\sum_{j=0}^{[(n_e-2k)/2]} \frac{(\lambda +n_e-2(k +j))}
{j! \,\Gamma(\lambda + n_e-2k +1 -j)} \, C^{(\lambda)}_{n_e-2(k+j)}(\omega_{s(e)}\cdot \omega_{t(e)}). 
$$

\smallskip

In the case where
$$ \beta^\R_{\Gamma,\mathbf{m},D,\ell_e} = \| x_{s(e)}-x_{t(e)} \|^{2\ell} 
( \log( \frac{m_e  \| x_{s(e)}-x_{t(e)} \|}{2} ) -\frac{\psi(\ell_e+1)+\psi(\lambda +\ell_e+1)}{2}) $$
we take the Taylor expansion
$$ \rho_e^{2\ell_e} 
(1-2 \frac{r_e}{\rho_e}  x + (\frac{r_e}{\rho_e})^2)^{\ell_e} = $$
$$ \rho_e^{2\ell_e}  \sum_{n=0}^{\ell_e} \frac{\Gamma(\ell_e+1)}{\Gamma(\ell_e-n)} (-1)^n
(2 \frac{r_e}{\rho_e})^n (1+ (\frac{r_e}{\rho_e})^2)^{\ell_e-n}\, x^n, $$
which, combined with \eqref{tmGegen} gives
$$ \| x_{s(e)}-x_{t(e)} \|^{2\ell}  = \rho_e^{2\ell_e} 
(1-2 \frac{r_e}{\rho_e} \omega_{s(e)}\cdot \omega_{t(e)}  + (\frac{r_e}{\rho_e})^2)^{\ell_e} = $$
$$ \rho_e^{2\ell_e} \sum_{n=0}^{\ell_e} 
\frac{\Gamma(\ell_e+1)\Gamma(\lambda) n! }{\Gamma(\ell_e-n)} (-1)^n
(\frac{r_e}{\rho_e})^n (1+ (\frac{r_e}{\rho_e})^2)^{\ell_e-n}\, 
 \sum_{k=0}^{[n/2]} \frac{(\lambda +n - 2k)}
{k! \,\Gamma(\lambda + n +1 -k)} \, C^{(\lambda)}_{n-2k}(\omega_{s(e)}\cdot \omega_{t(e)}). $$
We also use the expansion \eqref{ChebGenerate} into Chebyshev polynomials and
the expression of Chebyshev polynomials in terms of Gegenbauer polynomials \eqref{clamnm}.
This gives 
$$  \log( \frac{m_e  \| x_{s(e)}-x_{t(e)} \|}{2} )= \log(m_e)-\log(2) +\log( \| x_{s(e)}-x_{t(e)} \|) $$
where the first terms contribute coefficients in $\Q[\log(m_e), \log(2)]$ to the expansion
and the last term is expanded as
$$ \log( \| x_{s(e)}-x_{t(e)} \| )=\log(\rho_e) + \log((1+(\frac{r_e}{\rho_e})^2 -2 \frac{r_e}{\rho_e}
\omega_{s(e)}\cdot \omega_{t(e)})^{1/2}) = $$
$$ \log(\rho_e) - \sum_{n=1}^\infty T_n(\omega_{s(e)}\cdot \omega_{t(e)}) 
\frac{1}{n} (\frac{r_e}{\rho_e})^n = $$
$$  \log(\rho_e) - \sum_{n=1}^\infty  \frac{1}{n} (\frac{r_e}{\rho_e})^n
\sum_{m=0}^n  c^\lambda_{n,m} \, C^{(\lambda)}_m
(\omega_{s(e)}\cdot \omega_{t(e)}) . $$
As one can see from \eqref{tmGegen}, the coefficients $c^\lambda_{n,m}$ are in $\Q$
when $\lambda$ is an integer. When $\lambda$ is a 
half-integer, the coefficients contain values of the Gamma function
at half integers. By the duplication formula $\Gamma(z)\Gamma(z+\frac{1}{2}) = 2^{1-2z} \sqrt{\pi} \Gamma(2z)$, the values of $\Gamma(z)$ at half integers are in $\Q[\sqrt{\pi}]$, hence the
$c^\lambda_{n,m}$ are in $\Q[\sqrt{\pi}^{\pm 1}]$.

\smallskip

Finally, observe that the products of Gegenbauer polynomials that we obtain when
we combine the expansion of $ \| x_{s(e)}-x_{t(e)} \|^{2\ell}$ with that of
$\log( \| x_{s(e)}-x_{t(e)} \| )$ in the product term, can in turn be expressed in terms
of a sum of Gegenbauer polynomials. In fact, using \eqref{Gegenexpl} we have
$$ C^{(\lambda)}_n(x) \, C^{(\lambda)}_m(x) =
\sum_{k=0}^{[n/2]} \sum_{j=0}^{[m/2]} (-1)^{j+k} 
 \frac{\Gamma(\lambda + n -k) \Gamma(\lambda + m -j)}
 {k! j! (n-2k)! (m-2j)! \Gamma(\lambda)^2} \, (2 x)^{n+m-2(k+j)}
$$
Setting
$$ \alpha_{r,n,m}^\lambda = \frac{(-1)^r 2^{m+n-2r}}{\Gamma(\lambda)^2} 
2^{-(n+m-2r)}  \Gamma(\lambda) (n+m-2r)! 
\sum_{j+k=r} \frac{\Gamma(\lambda + n -k) \Gamma(\lambda + m -j)}
 {k! j! (n-2k)! (m-2j)! } $$
 $$ \beta_{r,k,n,m}^\lambda =\frac{(\lambda +n+m-2(r+k))}
{k! \,\Gamma(\lambda + n+m-2r +1 -k)} $$
we obtain
$$  C^{(\lambda)}_n(x) \, C^{(\lambda)}_m(x) =\sum_{r=0}^{[(m+n)/2]} \alpha_{r,n,m}^\lambda 
\sum_{k=0}^{[(n+m-2r)/2]} 
 \beta_{r,k,n,m}^\lambda
 \, C^{(\lambda)}_{n+m-2(r+k)}(x)
$$ 

\smallskip

When $\lambda$ is an integer,  the values $\Gamma(\lambda+n)$, $n\in \N$ of the 
Gamma function are integers, while the values of the Digamma function
$\psi(\lambda +n)$, $n\in \N$ are of the form $H_{\lambda+n-1} - \gamma$ when $\lambda$
is an integer, where $H_n$ is the $n$-th harmonic number and $\gamma$ is the  
Euler--Mascheroni constant. 
Thus, the coefficients of the expansion of the Feynman amplitude
obtained in this way are in $\Q[m_e, \log(m_e), (2\pi)^{-1}, \gamma]$. When $\lambda$ is a 
half-integer, we also have coefficients that have, in the numerator or denominator,
values at half-integers of the Gamma function, and which are therefore in $\Q[\sqrt{\pi}^{\pm 1}]$.
Moreover, when $\lambda$ is a half-integer, there are also
coefficients in the expansion that contain values of the Digamma function 
at half-integers and these are of the
form $\psi(n+\frac{1}{2})=-\gamma - 2 \log(2) + \sum_{k=1}^n \frac{2}{k-1}$, 
in $\Q[\gamma, \log(2)]$.  Therefore, the  coefficients must be in 
$\Q[m_e, \log(m_e), \sqrt{\pi}^{\pm 1},\gamma, \log(2)]$. 
\endproof

\medskip

We can then apply the same computational method described in \S 4 of \cite{CeyMar2}
to the terms of this expansion, using the decomposition of the chain of integration 
$\sigma_{\Gamma,\R}$ into the chain described in \cite{CeyMar2} associated to acyclic
orientations of the graph. That leads to subdividing the integrals into an angular part,
that can be expressed in terms of integrals of products of Gegenbauer polynomials
at the star of half-edges around each vertex of the graph,
followed by an integrals in the radial coordinates of the resulting function
depending on the matching pairs of half edges. The method
follows the same patter described in \S 4 of \cite{CeyMar2} and we do not discuss it
detail here.  Clearly, although in principle explicitly computable by this
method, computations of Feynman integrals performed with this $x$-space method
become very quickly intractable as the complexity of the graph increases. Thus, we turn
instead to a different method, based on using the complexified form of the Feynman
amplitude \eqref{omegamZ} and algebro-geometric techniques using the wonderful
compactification spaces $\cX_{\cH}[n]$.

\section{Complex case: forms with logarithmic poles and Rota--Baxter renormalization}\label{RBSec}

In this section we consider the complex case of Feynman amplitudes, as
introduced in \S \ref{SecRulesComplex}. We show that we can fit this case
in the general framework of {\em algebraic renormalization}, based on
the Connes-Kreimer Hopf algebra of Feynman graphs and a Rota--Baxter
algebra, which we construct using complexes of differential forms
with logarithmic poles on the wonderful compactifications $\cX_{\cH}[n]$.
The renormalization procedure for Feynman amplitudes in this setting
consists of the following steps:
\begin{enumerate}
\item {\em Pullback} of the form $\omega_{\Gamma,\C}$ to $\cX_{\cH}[n]$,
under the map $\pi_n: \cX_{\cH}[n]\to \cX^n$ of the iterated blowup
construction.
\item {\em Regularization} of $\pi^*_n(\omega_{\Gamma,\C})$ by replacing
this closed form with a cohomologous form $\eta_\Gamma$ on $\cX_{\cH}[n]$
with logarithmic poles along the boundary divisor $\cD_n$ of the compactification
$\cX_{\cH}[n]$.
\item Construction of a {\em Rota--Baxter algebra} $\cR$ of forms on $\cX_{\cH}[n]$
with logarithmic poles along $\cD_n$ with a linear operator of ``polar subtraction"
defined in terms of Poincar\'e residues.
\item An {\em algebraic Feynman rule} given by a commutative algebra
homomorphism $\phi: \cH \to \cR$ from the Connes-Kreimer Hopf algebra of Feynman graphs
$\fH$ to the Rota--Baxter algebra, given by $\phi(\Gamma)=\eta_\Gamma$.
\item {\em Renormalization} of the Feynman amplitude by Birkhoff factorization of
this algebraic Feynman rule, using the coproduct and antipode of $\fH$ and
the Rota--Baxter operator of $\cR$.
\item {\em Evaluation} of the renormalized Feynman amplitude on the chain
$\tilde\sigma_{\Gamma,\C}=\pi_n^{-1}(\sigma_{\Gamma,\C})$.
\end{enumerate}
We also show that, in this setting we obtain the beta function of
the renormalization group as an element in the Lie algebra $\fg(\cR)$
of the group $G(\cR)={\rm Hom}_{\rm Alg}(\fH,\cR)$.

\medskip
\subsection{Algebraic renormalization}

A general algebraic framework for renormalization was formulated in
terms of Rota--Baxter algebras and Hopf algebras in \cite{EGK},
following the analysis of the BPHZ renormalization procedure
carried out in \cite{CK}. We recall here the setting, as given in \cite{EGK}.

\smallskip
\subsubsection{Rota--Baxter algebras}
A Rota--Baxter algebra of weight $\lambda=-1$ is
a commutative unital algebra $\cR$ with a linear operator
$T:\cR \to \cR$ satisfying the Rota--Baxter relation
\begin{equation}\label{RBop}
 T(x) T(y) = T(x T(y))+ T(T(x)y) +\lambda T(xy) .
\end{equation}

The main example that corresponds to the use of dimensional
regularization of Feynman integrals in \cite{CK} is the
algebra of Laurent series with $T$ given by the projection onto
the polar part.
 
The Rota--Baxter operator $T$ determines a splitting of $\cR$ into
a piece $\cR_+=(1-T)\cR$ and a piece $\cR_-$ which is the 
unitization of $T\cR$. Both the $\cR_\pm$ are algebras, due to
the Rota--Baxter relation \eqref{RBop}.

\smallskip
\subsubsection{Connes--Kreimer Hopf algebra of Feynman graphs}

We consider here the usual Connes--Kreimer Hopf algebra $\fH$ 
of Feynman graphs \cite{CK}. This is a free commutative algebra 
with generators the 1PI Feynman graphs $\Gamma$ of the theory,
graded by the number of internal lines, with
$$ \deg(\Gamma_1\cdots\Gamma_n)=\sum_i \deg(\Gamma_i), 
\ \ \ \deg(1)=0, $$
and with coproduct
$$ \Delta(\Gamma)=\Gamma\otimes 1 + 1 \otimes \Gamma + 
\sum_{\gamma \in \cV(\Gamma)}  \gamma \otimes \Gamma/\gamma $$
where the sum is over all subgraphs $\gamma$ such that the
quotient $\Gamma/\gamma$ is still a 1PI Feynman graph of the
same theory. The antipode is defined inductively by
$$ S(X)=-X -\sum S(X') X'' $$
for $\Delta(X)=X\otimes 1 + 1 \otimes X + \sum X' \otimes X''$, with
the $X'$ and $X''$ of lower degrees. 

\smallskip
\subsubsection{Algebraic Feynman rules}

In this algebraic setting a Feynman rule is simply a morphism
{\em of commutative algebras} from a commutative Hopf algebra $\fH$ 
to a Rota--Baxter algebra $\cR$ of weight $-1$,
\begin{equation}\label{phiFeyn}
\phi \in {\rm Hom}_{{\rm Alg}}(\fH, \cR).
\end{equation}

Notice that the morphism $\phi$ only preserves the commutative
algebra structure and has no required compatibility with the 
Hopf and the Rota--Baxter structures.

\smallskip
\subsubsection{Birkhoff factorization}

As was originally shown by Connes and Kreimer in \cite{CK} and
generalized to the Rota--Baxter setting in \cite{EGK}, given an
algebraic Feynman rule $\phi \in {\rm Hom}_{{\rm Alg}}(\fH, \cR)$
as above, the Hopf structure on $\fH$ together with the Rota--Baxter
structure on $\cR$ determine a Birkhoff factorization 
\begin{equation}\label{Birkhoff}
 \phi = (\phi_-\circ S)\star \phi_+ 
\end{equation}
into {\em commutative algebra homomorphisms} 
$$\phi_\pm \in  {\rm Hom}_{\rm Alg}(\fH,\cR_\pm).$$
The product in \eqref{Birkhoff} is the product in the affine
group scheme dual to the commutative Hopf algebra,
defined by $\phi_1\star \phi_2(X) = \langle \phi_1\otimes \phi_2, \Delta(X)\rangle$.
The Birkhoff factorization is constructed inductively via the Connes--Kreimer
formula
$$ \phi_-(X)=-T (\phi(X) +\sum \phi_-(X') \phi(X'')) $$
$$ \phi_+(X)=(1-T)(\phi(X) +\sum \phi_-(X') \phi(X'')) $$
where $\Delta(X)=1\otimes X + X \otimes 1 + \sum X'\otimes X''$.

\medskip
\subsection{Forms with logarithmic poles and Rota--Baxter structure}

We describe here how to obtain a Rota--Baxter structure for renormalization
in configuration spaces, based on forms with logarithmic poles on the
wonderful compactifications $\cX_\cH[n]$ and iterated Poincar\'e residues.

\smallskip
\subsubsection{Iterated Poincar\'e residues}

Let $\cD$ be a normal crossings divisor in a smooth variety $\cY$
and let $\Omega^\bullet_\cY(\log \cD)$ denote the sheaf of differential
forms on $\cY$ with logarithmic poles along $\cD$, as in \cite{Del}.
Let $\{ D_i \}_{i\in I}$ denote the components of $\cD$ and, for $J\subseteq I$,
$J=\{ j_1,\ldots, j_r \}$, let $\cD_J = D_{j_1}\cap \cdots \cap D_{j_r}$ denote
the intersection. 

Given a form $\eta \in \Omega^k_\cY(\log \cD)$ and a 
nonempty intersection $\cD_J$, there is an iterated Poincar\'e residue
${\rm Res}_{\cD_J}(\eta)$, which is a form in $\Omega^{k-r}_{\cD_J}$.
The pairing of ${\rm Res}_{\cD_J}(\eta)$ with a $(k-r)$-cycle $\Sigma$ in
$\cD_J$ is given by
\begin{equation}\label{LerayPair}
\int_\Sigma {\rm Res}_{\cD_J}(\eta) = \frac{1}{(2\pi i)^r} \int_{\cL_{\cD_J}(\Sigma)} \eta, 
\end{equation}
where $\cL_{\cD_J}(\Sigma)$ is the iterated Leray coboundary of $\Sigma$, which is 
a $k$-cycle in $\cY$ given by a $T^r$-torus bundle over $\Sigma$ obtained as 
composition of the Leray coboundary maps 
\begin{equation}\label{LerayComp}
\cL_{\cD_J}(\Sigma) =\cL_{D_{j_1}}\circ \cdots \circ \cL_{D_{j_k}} (\Sigma)
\end{equation}
of the components of 
$\cD_J$, see Proposition 5.14 of \cite{CeyMar2}. 

\smallskip

\begin{lem}\label{residueless}
Let $f_1, \cdots, f_n$ be local parameters such that the components $D_j$ of the
divisor $\cD_n$ are defined by $f_j=0$. 
Given $\eta\in \Omega^m_\cY(\log \cD)$, set
\begin{equation}\label{etaJ}
T(\eta) = \sum_{j=1}^n \frac{df_j}{f_j} \wedge {\rm Res}_{D_j}(\eta).
\end{equation}
Then, for any $\cD_J=D_{j_1}\cup\cdots \cup D_{j_r}$ for a subset $J\subset I$
of components of $\cD_n$, the form $T(\eta)$ satisfies 
\begin{equation}\label{TetaRes}
{\rm Res}_{\cD_J}(T(\eta)) = {\rm Res}_{\cD_J}(\eta).
\end{equation}
\end{lem}

\proof First observe that for the ordinary Poincar\'e residue for a single component
$D_j$ we have ${\rm Res}_j(T(\eta))={\rm Res}_j(\eta)$ by construction. Then, using
\eqref{LerayPair} and \eqref{LerayComp} we see that the iterated residues also agree,  
$$ \langle {\rm Res}_{\cD_J}(T(\eta)), \Sigma \rangle =
\langle T(\eta), \cL_{\cD_J}(\Sigma) \rangle $$ $$ =\langle {\rm Res}_{j_1} (T(\eta)),
\cL_{j_2}\cdots \cL_{j_k}(\Sigma) \rangle $$ $$ = \langle {\rm Res}_{j_1} (\eta),
\cL_{j_2}\cdots \cL_{j_k}(\Sigma) \rangle = \langle  {\rm Res}_{\cD_J}(\eta), \Sigma \rangle. $$
\endproof

Given a form $\eta \in \Omega^\bullet_\cY(\log \cD)$, we define the {\em polar subtraction}
$\eta_{\cD}$ to be
\begin{equation}\label{prepform}
\eta_{\cD} := \eta - T(\eta).
\end{equation}
We will show below that this operation of polar subtraction has the same
formal properties of other polar subtraction methods used in quantum
field theory (such as DimReg with MS), and can be treated similarly to
obtain a good notion of {\em renormalized} Feynman integrals in
configuration spaces. The result of the polar subtraction is a form $\eta_{\cD}$
for which all the iterated Poincar\'e residues vanish.

\begin{cor}\label{nores}
The form $\eta_\cD$ has vanishing Poincar\'e residues, ${\rm Res}_{\cD_J}(\eta_\cD)=0$ for all $J\subset I$.
\end{cor}

\proof This is an immediate consequence of Lemma \ref{residueless}.
\endproof

\smallskip
\subsubsection{Rota--Baxter algebra of the wonderful compactifications}

We now consider the wonderful compactifications $\cX_\cH[n]$ described
in Section \ref{ConfigSec} above, with the boundary divisor
\begin{equation}\label{Divisor}
\cD_n= \cX_{\cH}[n] \smallsetminus \left( (\cX \setminus \cH)^n \setminus  
( \bigcup_{ \{i,j\} \subset \underline{n}} \Delta_{\{i,j\}}) \right)
\end{equation}

Let $\Omega^{\rm even}_\cY(\log \cD)$ be the {\em commutative} algebra of
even differential forms on a variety $\cY$ with logarithmic poles along  a
normal crossings divisor $\cD$. 

\medskip

Define the linear operator $T: \Omega^{\rm even}_\cY(\log \cD) \to 
\Omega^{\rm even}_\cY(\log \cD)$ given as in \eqref{etaJ} by
\begin{equation}\label{RBlogpoles}
T : \eta \mapsto T(\eta)=\sum_{j=1}^n \frac{df_j}{f_j} \wedge {\rm Res}_{D_j}(\eta). 
\end{equation}

\medskip

\begin{thm}\label{RBthm}
The pair $(\cR, T)$, with $\cR=\Omega^{\rm even}_\cY(\log \cD)$ and with 
$T$ as in \eqref{RBlogpoles}, is a Rota--Baxter algebra of weight $-1$.
\end{thm}

\proof First observe that the operator $T$ is idempotent: in fact, by Lemma \ref{residueless},
we have 
\begin{equation}\label{TTeta}
T(T(\eta))=\sum_j \frac{df_j}{f_j} \wedge {\rm Res}_{D_j}(T(\eta))=T(\eta).
\end{equation}
We need to check that the Rota--Baxter relation \eqref{RBop} is
satisfied. We have
$$ T(\eta)\wedge T(\xi) = \sum_j \frac{df_j}{f_j} \wedge {\rm Res}_{D_j} (\eta)\wedge
\sum_k \frac{df_k}{f_k} \wedge {\rm Res}_{D_k}(\xi), $$
which can equivalently be written as
\begin{equation}\label{TxTy}
T(\eta)\wedge T(\xi) = \sum_{j<k} \frac{df_j}{f_j} \wedge \frac{df_k}{f_k} \wedge ({\rm Res}_{D_k}(\eta) \wedge {\rm Res}_{D_j}(\xi) - {\rm Res}_{D_j}(\eta)\wedge {\rm Res}_{D_k}(\xi)).
\end{equation}
For the term $T(\eta \wedge \xi)$ we obtain
$$ T(\eta \wedge \xi) = \sum_j \frac{df_j}{f_j} \wedge {\rm Res}_{D_j} (\eta\wedge \xi)  $$
$$ = \sum_j \frac{df_j}{f_j} \wedge {\rm Res}_{D_j} ((T(\eta)+(\eta-T(\eta)))\wedge 
(T(\xi)+(\xi-T(\xi)))) $$
$$ = \sum_j \frac{df_j}{f_j} \wedge {\rm Res}_{D_j} (T(\eta)\wedge T(\xi)) $$
$$ + \sum_j \frac{df_j}{f_j} \wedge {\rm Res}_{D_j} (T(\eta)\wedge (\xi-T(\xi))) $$
$$ + \sum_j \frac{df_j}{f_j} \wedge {\rm Res}_{D_j} ((\eta-T(\eta))\wedge T(\xi)) $$
$$ + \sum_j \frac{df_j}{f_j} \wedge {\rm Res}_{D_j} ((\eta-T(\eta))\wedge (\xi-T(\xi))). $$
The last term has a wedge of two forms with no residues on all the $D_j$; in the
two intermediate terms the forms $\eta_\cD=\eta-T(\eta)$ and $\xi_\cD=\xi - T(\xi)$
have no residue, hence using \eqref{TTeta} we obtain, respectively,
$T(\eta)\wedge (\xi-T(\xi)$ and $(\eta-T(\eta))\wedge T(\xi)$. Since the forms $\eta$
and $\xi$ are even, commuting them with a 1-form $df/f$ does not give rise to 
a non-trivial sign. Finally, the first term can be computed using the explicit form
of $T(\eta)$ and $T(\xi)$ and gives again the expression \eqref{TxTy}.
Thus, we have
\begin{equation}\label{Txyside}
T(\eta \wedge \xi) =T(\eta)T(\xi)+ T(\eta)\wedge (\xi-T(\xi) + (\eta-T(\eta))\wedge T(\xi).
\end{equation} 
We then compute the remaining terms $T(\eta \wedge T(\xi))$ and $T(T(\eta)\wedge \xi)$.
Applying \eqref{Txyside} to these cases and using \eqref{TTeta} we obtain
$$ T(\eta \wedge T(\xi)) = T(\eta)\wedge T(\xi) + (\eta-T(\eta)) \wedge T(\xi) $$
$$ T(T(\eta)\wedge \xi) = T(\eta)\wedge T(\xi) +T(\eta)\wedge (\xi-T(\xi)), $$
so that we obtain the Rota--Baxter relation of weight $-1$ as stated,
$$  T(\eta \wedge T(\xi)) + T(T(\eta)\wedge \xi) - T(\eta)\wedge T(\xi) = T(\eta\wedge \xi) . $$
\endproof

\smallskip

To apply this construction to our renormalization problem, we still need to take
into account the fact that Feynman graphs with different number of edges 
$n=|E_\Gamma|$ correspond to forms $\eta_\Gamma$ with logarithmic poles 
in different varieties $\cX_{\cH}[n]$ with the divisors $\cD_n$ of  \eqref{Divisor}. 
Thus, we need to combine
the Rota--Baxter algebras $(\Omega^{{\rm even}}_{\cX_{\cH}[n]}(\log \cD_n), T_n)$
constructed as in Theorem \ref{RBthm} into a single Rota--Baxter algebra where all
values of $n$ can occur. 

\smallskip

Let $\bigwedge_n \Omega^{\rm even}_{\cX_\cH[n]}(\log \cD_n)$ denote the
algebra generated by elements of the form $\eta_{n_1}\wedge\cdots \wedge \eta_{n_k}$
for arbitrary $n_1,\ldots, n_k$.
We then define the Rota--Baxter algebra of weight $-1$ for the renormalization
problem in the following way.

\begin{thm}\label{RBalgalln}
Consider the commutative algebra
\begin{equation}\label{Ralln}
\cR= \bigwedge_n \Omega^{\rm even}_{\cX_\cH[n]}(\log \cD_n),
\end{equation}
as above. Let $T: \cR \to \cR$ be the linear operator defined by setting
\begin{equation}\label{Tprod}
T(\eta_1\wedge \eta_2)= T_{n_1}(\eta_1)\wedge \eta_2 + \eta_1 \wedge T_{n_2}(\eta_2)
- T_{n_1}(\eta_1)\wedge T_{n_2}(\eta_2),
\end{equation}
for $\eta_i \in \Omega^{\rm even}_{\cX_\cH[n_i]}(\log \cD_{n_i})$.
Then $(\cR,T)$ is a Rota--Baxter algebra of weight $-1$. 
\end{thm}

\proof We first check that \eqref{Tprod} determines a Rota--Baxter structure of
weight $-1$ on $$\Omega^{\rm even}_{\cX_\cH[n_1]}(\log \cD_{n_1})\wedge \Omega^{\rm even}_{\cX_\cH[n_2]}(\log \cD_{n_2})$$ and then inductively that it determines a linear operator $T$
of $\cR$ that also satisfies the Rota--Baxter relation of weight $-1$.
By \eqref{Tprod} we have
$$ T(T_{n_1}(\eta_1)\wedge \eta_2)= T_{n_1}(\eta_1)\wedge \eta_2 +
T_{n_1}(\eta_1)\wedge T_{n_2}(\eta_2) - T_{n_1}(\eta_1)\wedge T_{n_2}(\eta_2) = 
T_{n_1}(\eta_1)\wedge \eta_2 $$
$$ T(\eta_1\wedge T_{n_2}(\eta_2))= T_{n_1}(\eta_1)\wedge T_{n_2}(\eta_2)
+ \eta_1 \wedge T_{n_2}(\eta_2)   - T_{n_1}(\eta_1)\wedge T_{n_2}(\eta_2) = 
\eta_1 \wedge T_{n_2}(\eta_2), $$
hence the Rota--Baxter relation is satisfied. One easily checks also that
the operator $T$ defined in this way still satsifies $T^2=T$.
Now assume inductively that, on all $k$-fold wedge products
$$ \Omega^{\rm even}_{\cX_\cH[n_1]}(\log \cD_{n_1})\wedge \cdots \wedge \Omega^{\rm even}_{\cX_\cH[n_k]}(\log \cD_{n_k}) $$ 
we have constructed $T$ satisfying the Rota--Baxter relation of weight $-1$, with $T^2=T$.
Then we can extend it to a $(k+1)$-fold wedge product by setting 
$$ T(\eta_{n_1}\wedge\cdots \wedge \eta_{n_k}\wedge \eta_{n_{k+1}}) =
T(\eta_{n_1}\wedge\cdots \wedge \eta_{n_k}) \wedge \eta_{n_{k+1}} $$ $$
+ \eta_{n_1}\wedge\cdots \wedge \eta_{n_k} \wedge T_{n_{k+1}}(\eta_{n_k}) 
- T(\eta_{n_1}\wedge\cdots \wedge \eta_{n_k}) \wedge T_{n_{k+1}}(\eta_{n_k}) . $$
The sam argument above shows that it still satisfied the Rota--Baxter relation and
one sees in the same way that $T^2=T$.
\endproof

We refer to the Rota--Baxter algebra obtained as in Theorem \ref{RBalgalln} as
the {\em Rota--Baxter algebra of configuration spaces}. 

\medskip
\subsection{Algebraic renormalization of Feynman integrals in configuration space}

\smallskip
\subsubsection{The algebraic Feynman rule}

Let $\omega_{\Gamma,\C}$ be the {\em massless} Feynman amplitude 
for a Feynman graph $\Gamma$ in the complex case, defined as in \eqref{Comegam0},
as a smooth closed form in $\Omega^{2 Dn}(\cF(\cX\smallsetminus \cH, n))$. As in
\cite{CeyMar2}, we then consider the pullback $\pi_n^* (\omega_{\Gamma,\C})$ under
the iterated blowup map $\pi_n: \cX_\cH[n] \to \cX^n$. The pullback form is a smooth
closed form in $\Omega^{2 Dn}(\cX_\cH[n]\smallsetminus \cD_n)$, where $\cD_n$
is the boundary divisor \eqref{Divisor}. 

\smallskip

It is well known that the de Rham cohomology of a smooth quasi-projective 
varieties computed using algebraic differential forms, \cite{Groth}. Moreover,
if the variety is the complement of a normal crossings divisor one can use algebraic
differential forms with logarithmic poles along the divisors, \cite{Del}, so that we have
$$ H^*(\cX_\cH[n]\smallsetminus \cD_n) \simeq {\mathbb H}^*(\cX_\cH[n], 
\Omega_{\cX_\cH[n]}^*(\log(\cD_n))). $$
Thus, the smooth closed differential form $\pi_n^* (\omega_{\Gamma,\C})$
is cohomologous to an algebraic differential form $\eta_\Gamma$ with
logarithmic poles along $\cD_n$,
$$ [ \pi_n^* (\omega_{\Gamma,\C}) ] = [\eta_\Gamma] \in H^{2Dn}(\cX_\cH[n]\smallsetminus \cD_n). $$

\smallskip

Let $\cR$ be the Rota--Baxter algebra of configuration spaces as in Definition \ref{RBalgalln}.
We define the algebraic Feynman rule $\phi \in {\rm Hom}_{\rm Alg}(\fH,\cR)$, 
for the algebraic renormalization of Feynman integrals in configuration spaces, by setting
\begin{equation}\label{philog}
\phi(\Gamma_1 \cdot \cdots \cdot 
\Gamma_k) = \eta_{\Gamma_1} \wedge \cdots \wedge \eta_{\Gamma_k},
\end{equation}
where for Feynman graph $\Gamma$ the form 
$\eta_\Gamma  \in \Omega^{2D n}(\cX_{\cH}[n], \cD_n)$, with $n=|E_\Gamma|$,  
is the algebraic differential form with logarithmic poles along the divisor 
$\cD_n$ of \eqref{Divisor}, obtained as above.

\smallskip
\subsubsection{The Birkhoff factorization}

The resulting Birkhoff factorization of the Feynman rule \eqref{philog} determines
$\phi_-$ and $\phi_+$ inductively as
$$ \phi_-(\Gamma)=-T (\eta_\Gamma +\sum_{\gamma\subset \Gamma} 
\phi_-(\gamma) \wedge \eta_{\Gamma/\gamma} ) $$
$$ \phi_+(\Gamma)=(1-T)(\eta_\Gamma +\sum_{\gamma\subset \Gamma} 
\phi_-(\gamma) \wedge \eta_{\Gamma/\gamma} ) 
 = \eta_{\Gamma,\cD} + \sum_{\gamma\subset \Gamma} 
(\phi_-(\gamma) \wedge \eta_{\Gamma/\gamma} )_{\cD}, $$
with $\eta_{\cD}$ the polar subtraction defined as in \eqref{prepform}, and
similarly for the other terms.
By analogy to the usual setting of BPHZ renormalization (see \cite{CK}), we refer
to the expression $\eta_\Gamma +\sum_{\gamma\subset \Gamma} \phi_-(\gamma) \wedge \eta_{\Gamma/\gamma}$ as the {\em preparation} of the form $\eta_\Gamma=\phi(\Gamma)$,
to the $\phi_-(\Gamma)$ as the {\em counterterms}, and to the $\phi_+(\Gamma)$ as the
renormalized Feynman amplitude in configuration space.

\smallskip
\subsubsection{Renormalized Feynman integrals}

One then obtains, as renormalized Feynman integral in configuration space the
integral
\begin{equation}\label{renormint}
\int_{\tilde\sigma_{\Gamma,\C}} \eta_{\Gamma,\cD} + \sum_{\gamma\subset \Gamma} 
(\phi_-(\gamma) \wedge \eta_{\Gamma/\gamma} )_{\cD},
\end{equation}
where $\tilde\sigma_{\Gamma,\C}=\pi_n^{-1}(\sigma_{\Gamma,\C})=
\P^D_{\cH}[n]\times \{ y \}$ is the preimage in $\cX_\cH[n]$ of the chain of
integration $\sigma_{\Gamma,\C}$ of the unrenormalized Feynman
amplitude \eqref{WeightZ} in the complex case. The form 
$\eta_{\Gamma,\cD} + \sum_{\gamma\subset \Gamma}  
(\phi_-(\gamma) \wedge \eta_{\Gamma/\gamma} )_{\cD}$ now has vanishing
residues on all the divisors $\cD_J$, hence the integral \eqref{renormint}
is free of divergences.

\begin{lem}
If the motives $\m(\cX)$ and $\m(\cH)$ of the varieties $\cX$ and $\cH$ are mixed Tate,
then the motive $\m(\cX_\cH[n])$ of the configuration space is also mixed Tate, and so
are the motives $\m(\cD)$ of all the unions and intersections $\cD$ of components of the
boundary divisor $\cD_n$. In particular the mixed motive $\m(\cX_\cH[n],\cD_n)$
is also mixed Tate.
\end{lem}

\begin{proof}
Due to Voevodsky (Proposition 3.5.3 in \cite{Vo}), the motive $m(\tilde{\cY})$ of  the 
blow-up of a smooth scheme $\cY$  along a smooth closed subscheme $\cV \subset \cY$ 
is canonically isomorphic to
$$
m(\cY) \oplus \bigoplus _{k=1}^{\text{codim}_\cY(\cV)-1} m(\cV)(k)[2k]
$$
where $[-]$ denotes the degree shift in the triangulated category of mixed motives, 
while $(-)$ is the Tate twist. The result for $\m(\cX_\cH[n])$ simply follows from the 
repeated application of this blowup formula to the iterated blowup construction 
of $\cX_\cH[n]$. 
The case of the motives $\m(\cD)$ follows  
from the explicit stratification of $\cX_\cH[n]$, see \cite{KS}, where the intersections
of the components are described explicitly in terms of other configuration spaces
whose motive can be computed by the same method. Knowing that all the intersections
are still mixed Tate motives then implies inductively that the unions also are, by
repeatedly using the distinguished triangles in the triangulated category of mixed motives.
\end{proof}

Thus, we can formulate the renormalized Feynman integral as a period, in
the sense of \cite{KZ}, in the following way.

\begin{cor}\label{periodren}
The renormalized Feynman integral \eqref{renormint} is a period of a mixed Tate motive
over the field of definition of the algebraic form $\phi_+(\eta_\Gamma)$.
\end{cor}

\proof The motive of the wonderful compactification $\cX_{\cH}[n]$ is mixed Tate
whenever the motive of $\cX$ is. This follows the same lines as the argument used in
\cite{CeyMar1} for the graph configuration spaces, using the description of
$\cX_{\cH}[n]$ as an iterated blowup of a mixed Tate variety along a mixed Tate locus
and the blowup formula for motives.
The varieties $\cX_{\cH}[n]$ are defined over $\Z$, so we only need
to worry about the algebraic differential form $\phi_+(\eta_\Gamma)$,
which may be defined over $\C$ or possibly a smaller field.
\endproof

\begin{rem}\label{periodsnature} {\rm 
This renormalization method does not provide a direct description
of the period integral \eqref{periodren} as a combination of multiple
zeta values with coefficients in $\Q[(2\pi i)^{-1}]$.  In fact, 
one knows by \cite{Brown} that periods of mixed Tate motives over $\Z$
would be of this form and that the varieties $\cX_{\cH}[n]$ are 
defined over $\Z$, but we have constructed the differential form $\eta_\Gamma$ as
a form cohomologous to $\pi_n^*(\omega_{\Gamma,\C})$ in
$H^*(\cX_{\cH}[n]\smallsetminus \cD_n,\C)$. Thus, a priori,
$\eta_\Gamma$ and its renormalization $\phi_+(\eta_\Gamma)$ 
are only $\C$-forms and not $\Q$-forms. A more detailed analysis of the
relation to multiple zeta values will be carried out elsewhere \cite{Cey}. }
\end{rem}

\smallskip 
\subsubsection{The massive case}

In this setting of regularization by algebraic forms with logarithmic poles
on $\cX_{\cH}[n]$ and renormalization by pole subtraction and Birkhoff
factorization of the regularized forms, one can treat the massive case
as a formal series, using an expansion such as \eqref{tayloromegaC}
$$ \omega_{\Gamma,\mathbf{m}, \C} = \sum_{\underline{\ell}}
\prod_{e\in E_\Gamma} \frac{B^\C_{\Gamma,\mathbf{m},D,\ell_e}}{(2\pi)^D}
 \,\,\, \beta^\C_{\Gamma,\mathbf{m},D,\ell_e} $$
or the asymptotic expansion as in \eqref{expandomega}
$$ 
\omega_{\Gamma,\mathbf{m}, \C} \sim \sum_{\underline{\ell}}
C^\C_{\Gamma,\mathbf{m},D,\underline{\ell}} \,\,\, \omega_{\Gamma,\mathbf{m},D,\underline{\ell},\C},
$$
One can then consider a regularization $\eta_{\Gamma,\mathbf{m},D,\ell_e}$
of the forms $\beta^\C_{\Gamma,\mathbf{m},D,\ell_e}$,
or, respectively, of the $\omega_{\Gamma,\mathbf{m},D,\underline{\ell},\C}$
and then proceed to the pole subtraction and renormalization of the forms  
$\eta_{\Gamma,\mathbf{m},D,\ell_e}$ according to the Rota--Baxter renormalization
procedure described in this section.

\smallskip
\subsection{The renormalization group}

In the Connes--Kreimer theory of renormalization based on Birkhoff factorization \cite{CK},
the beta function of the renormalization group equation is lifted to the Lie algebra of the
affine group scheme dual to the Hopf algebra of Feynman graphs. In the subsequent
work \cite{CoMa}, it is further lifted to the level of a universal
Hopf algebra, where the universal counterterms are obtained from a {\em universal singular
frame} depending on the beta function, seen as an element in a free Lie algebra. This algebraic
approach to the renormalization group flow was further developed in terms of Dynkin idempotents,
\cite{EGP}, see also \cite{Patras} for a survey of these recent results. 

\smallskip

Let $\cA$ be an associative algebra. 
the Dynkin operator $\cD_n \in \Q[S_n]$ with $S_n$ the symmetric group, is defined by
the property that
\begin{equation}\label{Dnop}
[\cdots [x_1,x_2],\cdots x_n] = x_1 x_2 \cdots x_n \, \cD_n,
\end{equation}
for any collection of elements $x_1,\ldots, x_n \in \cA$, with the symmetric group
acting on the right. One defines $\cD=\sum_{n\geq 0} \cD_n$.

\smallskip

Consider then the group of algebraic Feynman rules $G(\cR)={\rm Hom}_{{\rm Alg}}(\fH,\cR)$
and let $\fg(\cR)$ be the corresponding Lie algebra. It is shown in \cite{EGP} that right 
composition with the Dynkin operator $\cD$ gives a bijective map from $G(\cR)$ to
$\fg(\cR)$, whose inverse map is given by the ``universal frame for renormalization"
in the sense of \cite{CoMa},
\begin{equation}\label{unisingframe}
\fg(\cR) \ni \beta \mapsto \sum_{n\geq 0, k_i >0}
\frac{ \beta_{k_1} \star \cdots \star \beta_{k_n}}{k_1 (k_1+k_2)\cdots (k_1+\cdots+k_n)} \in G(\cR),
\end{equation}
where $\beta=\sum_n \beta_n$ with respect to the grading of the Hopf algebra.

\smallskip

In our setting we obtain an interpretation of the renormalization group by applying
this construction with $\cR$ the Rota--Baxter algebra of differential forms on the
wonderful configuration spaces with logarithmic poles along their boundary divisors.
Namely, one identifies the beta function of the renormalization group with that
element $\beta$ in the Lie algebra $\fg(\cR)$ that is the image under right
composition with the Dynkin operator $\cD$ of the negative piece of
the Birkhoff factorization 
$$ \phi_-(\Gamma)=-T (\eta_\Gamma +\sum_{\gamma\subset \Gamma} 
\phi_-(\gamma) \wedge \eta_{\Gamma/\gamma} ) $$
of the regularized Feynman amplitude $\eta_\Gamma$. Thus, $\beta \in \fg(\cR)$
is determined by the property that
\begin{equation}\label{Dphibeta}
\phi_- \cdot \cD  = \beta \ \ \text{ or equivalently } \ \ 
\phi_- = \sum_{n\geq 0, k_i >0}
\frac{ \beta_{k_1} \star \cdots \star \beta_{k_n}}{k_1 (k_1+k_2)\cdots (k_1+\cdots+k_n)}.
\end{equation}


\section{Beyond Scalar QFTs} \label{nonscalarSec}

The more realistic QFTs, such as gauge theories, involve particles  with spin or polarization and
the propagators for such fields are slightly different from the scalar case. We show that, 
nonetheless, these propagators have similar properties, which are sufficient to describe the
corresponding Feynman integrals, as in scalar case, as periods. 
 
\subsection{Spin 1/2 particles: the Dirac fermions} \label{diracSec}  
In Euclidean quantum electrodynamics, the propagator for a Dirac field representing the electron 
is the fundamental solution of  the Dirac equation
$$
(-i \slashed\partial + m) S(x) = \delta(x)
$$
where $\gamma^\mu$ are  the gamma matrices and $\slashed\partial =  \gamma^\mu \partial_\mu$.
The Green function $S(x)$ can be given as
\begin{eqnarray}\label{DiracG}
S_{m, D}(x ) =  \frac{1}{(2 \pi )^{D}} \int_{\R^{D} } dk \ 
\frac{e^{i k \cdot x }} {-  \slashed k + m} 
=  \frac{1}{(2 \pi )^{D}} \int_{\R^{D} } dk \ e^{i k \cdot x}
\frac{ \slashed k + m} {\| k \|^2 + m} .
\end{eqnarray}

\begin{lem} \label{diracLem}
The  Euclidean Dirac propagator on $\R^D$, with 
$D=2\lambda+2$,  satisfies
\begin{eqnarray*}
S_{m, D} (x) &=& \left(\frac{1}{2 \pi} \right)^{\lambda +1}  (i \gamma^\mu x_\mu)
\frac{m^{\frac{\lambda}{2}} }{\| x \|^{\lambda+1}}
\left(
 \frac{\lambda} {\| x \|} \cK_\lambda (\sqrt{m} \|x \|)
+  \frac{\sqrt{m}} {2 }  ( \cK_{\lambda-1} (\sqrt{m} \|x \|) +  \cK_{\lambda+1} (\sqrt{m} \|x \|))
\right) \\
&+&  \left(\frac{1}{2 \pi} \right)^{\lambda +1} \frac{m^{\frac{\lambda+2}{2}} }{\| x \|^{\lambda}}
\cK_\lambda (\sqrt{m} \|x \|)  .
\end{eqnarray*}
\end{lem}
  
\begin{proof}
A simple computation using \eqref{DiracG} gives the Dirac propagator 
$S_{m, D}$ on $\R^D$ in terms of the scalar
field propagator as follows:
$$
S_{m, D} = ( - i \slashed \partial + m) G_{\sqrt{m},D}.
$$
Then, we substitute the scalar propagator $G_{m,D}$ given in Lemma 
\ref{GmBessellem} and we use the identity 
$$
\frac{\partial}{\partial z} K_\lambda(z) = - \frac{1}{2} \left( \cK_{\lambda+1}(z) + \cK_{\lambda-1}(z)  \right). 
$$
\end{proof}  
    
\begin{cor}
The Dirac propagator $S_{m, D}(x)$ is smooth on $\R^D \setminus \{0\}$ with  $D=2 \lambda +2$.
Moreover, it diverges at $x=0$, and $S_{m, D}(x) \to 0 $ as $\| x \| \to \infty$. 
\end{cor}  

\subsection{Spin 1 particles: the gauge bosons} \label{bosonSec}    
The propagator of a spin 1 boson with mass $m$ is given by 
\begin{eqnarray}\label{bosonG}
\triangle^{\mu \nu}_{m, D}(x ) =     \frac{1}{(2 \pi )^{D}} 
\int_{\R^{D} } dk \ e^{i k \cdot x}
\left( 
\frac{g_{\mu\nu} + \frac{k_\mu k_{\nu}}{m^2}} {\| k \|^2 + m}
-  \frac{  \frac{k_\mu k_{\nu}}{m^2}} {\| k \|^2 + \frac{m}{\alpha}}
\right)
\end{eqnarray}
in the Stueckelberg gauge, depending on a gauge parameter $\alpha$.

\begin{lem} \label{bosonLem}
The propagator for a massive vector field can be derived from the scalar 
field propagator $G_{m,D}$:
\begin{eqnarray*}
\triangle^{\mu \nu}_{m, D}(x ) 
= g_{\mu\nu}  G_{\sqrt{m},D} (x) + \frac{1}{m^2} 
\left( 
\partial_\mu \partial_\nu G_{\sqrt{\frac{m}{\alpha}},D} (x) 
- \partial_\mu \partial_\nu G_{\sqrt{m},D} (x) 
\right).
\end{eqnarray*}
Therefore, the  propagator $\triangle^{\mu \nu}_{m, D}(x ) $ is smooth on $\R^D \setminus \{0\}$ with  
$D=2 \lambda +2$.
Moreover, it diverges at $x=0$, and $\triangle^{\mu \nu}_{m, D}(x )  \to 0 $ as $\| x \| \to \infty$. 
\end{lem}

In other words, the Dirac propagators and the massive gauge bosons propagators  
can be related to propagators for the scalar case,
hence any Feynman integrand involving an amplitude constructed out of these
two types of propagators can also be interpreted in terms of periods of 
mixed Tate motives,  following the arguments we developed for the scalar case.

\smallskip

\subsection*{Acknowledgments} 
The first author is partially supported by PCIG11-GA-2012-322154. 
The second author acknowledges support from NSF grants DMS-0901221, DMS-1007207, 
DMS-1201512, and PHY-1205440 and the hospitality and support of the Mathematical
Sciences Research Institute in Berkeley and of the Kavli Institute for Theoretical Physics China
and the Morningside Center for Mathematics in Beijing. The second author also thanks Li Guo for
useful conversations.


\begin{thebibliography}{10}

\bibitem{AdHed} D.R.~Adams, L.I.~Hedberg, {\em Function spaces and potential theory}, Springer

\bibitem{BDS} A.~Bezubik, A.~Dabrowska, A.~Strasburger, {\em On spherical expansions
of zonal functions on Euclidean spheres}, Arch. Math. 90 (2008) 70--81.

\bibitem{BK}  S.~Bloch, D.~Kreimer, {\em Feynman amplitudes and Landau singularities 
for one-loop graphs}, Commun. Number Theory Phys. 4 (2010), no. 4, 709--753. 

\bibitem{Boyd} W.G.C.~Boyd, {\em Stieltjes transforms and the Stokes phenomenon}, Proc.
R. Soc. Lond. A 429 (1990) 227--246.

\bibitem{Brown} F.~Brown, {\em Mixed Tate motives over $\Z$},
Ann. of Math. 175 (2012) N.2, 949--976.

\bibitem{Cey} \"O. Ceyhan, in preparation.

\bibitem{CeyMar1} \"O. Ceyhan, M. Marcolli {\it Feynman integrals and motives 
of configuration spaces}, Comm. Math. Phys. 313 (2012), no. 1, 35--70.

\bibitem{CeyMar2} \"O. Ceyhan, M. Marcolli {\it Feynman integrals and 
periods in configuration spaces},  preprint, arXiv:1207.3544.

\bibitem{CK} A.~Connes, D.~Kreimer, {\em Renormalization in quantum field theory 
and the Riemann-Hilbert problem. I. The Hopf algebra structure of graphs and 
the main theorem}, Comm. Math. Phys. 210 (2000), no. 1, 249--273.

\bibitem{CoMa} A.~Connes, M.~Marcolli, {\em Renormalization and motivic Galois
theory}, Int. Math. Res. Notices 76 (2004) 4073--4091. 

\bibitem{Del} P.~Deligne, {\it Th\'eorie de Hodge. II}, Inst. Hautes \'Etudes Sci. Publ. Math. 
No. 40 (1971), 5--57.

\bibitem{EGP} K.~Ebrahimi-Fard, J.~Gracia-Bond\`ia, F.~Patras, {\em 
A Lie theoretic approach to renormalization}, Comm. Math. Phys. 276 (2007), no. 2, 519--549.

\bibitem{EGK} K.~Ebrahimi-Fard, L.~Guo, D.~Kreimer, {\em Integrable renormalization. II. The general case}. Ann. Henri Poincar\'e 6 (2005), no. 2, 369--395.

\bibitem{FM}  W.~Fulton, R.~MacPherson, {\it A compactification of configuration spaces.} 
Ann. of Math. (2) 139 (1994), no. 1, 183--225.

\bibitem{Groth} A.~Grothendieck, {\it On the de Rham cohomology of algebraic varieties}, Inst. Hautes
\'Etudes Sci. Publ. Math. No. 29 (1966) 95--103.

\bibitem{KS} B.~Kim, F.~Sato, 
{\it A generalization of Fulton-MacPherson configuration spaces.}, 
Selecta Math. (N.S.) 15 (2009), no. 3, 435--443.

\bibitem{KZ} M.~Kontsevich, D.~Zagier, 
{\it Periods}, in ``Mathematics unlimited -- 2001 and beyond", 771--808, Springer, 2001. 
 
 \bibitem{Li1} L.~Li, {\it Chow motive of Fulton--MacPherson configuration spaces and wonderful 
 compactifications}, Michigan Math. J. 58 (2009),  N.2, 565--598.
 
 \bibitem{Li2} L.~Li, {\it Wonderful compactification of an arrangement of subvarieties}, 
 Michigan Math. J. 58 (2009), no. 2, 535--563.
 
 \bibitem{Martin} N.~Martin, {\em Compactification d'espaces de configurations}, M\'emoire
 de Master 2 Recherce, \'Ecole Polytechnique, 2013. 
 
 \bibitem{Patras} F.~Patras, {\em Dynkin operators and renormalization group actions in pQFT}, in ``Vertex operator algebras and related areas", 169--184, Contemp. Math., 497, Amer. Math. Soc., Providence, RI, 2009. 
 
 \bibitem{SW} E.~Stein, G.~Weiss, {\em Introduction to Fourier analysis on Euclidean spaces},
 Princeton, 1971.

\bibitem{Titch} E.C.~Titchmarsh, {\em An introduction to the theory of Fourier integrals}, Oxford, 1948.

\bibitem{Vo} V.~Voevodsky, {\em Triangulated categories of motives over a field.} 
in Cycles, transfer and motivic homology theories, pp. 188--238, Annals of Mathematical 
Studies, Vol. 143, Princeton, 2000.

\bibitem{Watson} G.N.~Watson, {\em A treatise in the theory of Bessel functions}, Cambridge, 1944.

\end{thebibliography}
\end{document}